\documentclass[11pt]{article}
\usepackage[letterpaper,margin=1in]{geometry}

\usepackage[T1]{fontenc}
\usepackage[utf8]{inputenc}
\usepackage{lmodern}
\usepackage{microtype}
\usepackage{amsmath,amssymb,mathtools,amsthm}
\usepackage{enumitem}
\usepackage{needspace}
\usepackage{xcolor}
\usepackage{etoolbox}
\usepackage{float}
\usepackage{tikz}
\usepackage{todonotes}

\newtheorem{theorem}{Theorem}[section]
\newtheorem{lemma}[theorem]{Lemma}
\newtheorem{proposition}[theorem]{Proposition}
\newtheorem{corollary}[theorem]{Corollary}
\theoremstyle{definition}
\newtheorem{definition}[theorem]{Definition}
\newtheorem{remark}[theorem]{Remark}

\newcommand{\hide}[1]{}
\newcommand{\F}{\mathbb F}
\newcommand{\ket}[1]{\lvert #1\rangle}
\newcommand{\bra}[1]{\langle #1\rvert}
\newcommand{\norm}[1]{\left\lVert #1\right\rVert}
\newcommand{\abs}[1]{\left\lvert #1\right\rvert}
\newcommand{\Tr}{\operatorname{Tr}}
\newcommand{\supp}{\operatorname{supp}}
\newcommand{\Id}{\mathbb I}
\newcommand{\Hcal}{\mathcal H}

\allowdisplaybreaks[2]
\setlist{leftmargin=*,topsep=.35em,itemsep=.12em,parsep=0pt,partopsep=0pt}
\usepackage{quantikz}
\usepackage{url}
\usepackage{hyperref}

\hypersetup{
  colorlinks=true,
  linkcolor=blue!55!black,
  citecolor=blue!55!black,
  urlcolor=blue!55!black,
  pdftitle={Random Garbage Creates Exponential Quantum Query Separations},
  pdfauthor={Khaled Elbassioni, Rishikesh Gajjala, Saurabh Ray}
}

\begin{document}

\title{Random Garbage Creates Exponential Quantum Query Separations}
\author{
  Khaled Elbassioni\textsuperscript{1}\quad
  Rishikesh Gajjala\textsuperscript{3}\quad
  Saurabh Ray\textsuperscript{2,3}\\[0.5em]
  \small\textsuperscript{1} Khalifa University of Science and Technology,
  Abu Dhabi, United Arab Emirates\\
  \small\textsuperscript{2} Division of Science, New York University Abu Dhabi, United Arab Emirates\\
  \small\textsuperscript{3} Center for Quantum and Topological Systems, NYUAD Research Institute, United Arab Emirates
}
\date{}

\maketitle

\begin{abstract}
We prove an exponential separation between standard (XOR) quantum queries
and forward-erasing queries for a Boolean property of injective functions,
answering a question of Aaronson~\cite{Aaronson2021Open}.  We also obtain
an exponential separation for a Boolean promise problem on permutations,
where forward in-place queries are unitary but inverse queries are not
supplied.

Our construction follows Aaronson's suggestion of appending fixed random
garbage to a Simon function.  Let $X=\{0,1\}^n$ and $N=2^n$.  We use the
injection $f_{h,r}(x)=(h(x),x,r_x)$, where $h:X\to X$ is promised to be
either a permutation or a Simon function, and $r:X\to X$ is an arbitrary
tag table that is reused on every query.  The resulting problem is
solvable with at most $n+2$ standard queries, but requires
$\Theta(\sqrt N)$ forward-erasing queries, matching classical collision
sampling.  Embedding these instances into permutations, with the copy of
$x$ replaced by a single bit that keeps the embedding injective, gives a
promise problem on permutations of $L=4N^2$ elements that requires
$O(\log L)$ standard queries versus $\Theta(L^{1/4})$ forward in-place
queries.  The lower bounds remain
valid when a control qubit coherently selects whether to apply or bypass a
query.

The lower bounds use the recording method.  We replace each query by a
nearby contraction that records cells of a purified tag table.  After $T$ queries, tracing out the purified tag table yields a sum of positive-semidefinite contributions,
each depending on $h$ at no more than $T$ addresses. A query-by-query
hybrid bounds the trace-norm difference between the actual and auxiliary
final operators by $O(T/\sqrt N)$.  Random permutations and random Simon
functions have the same restriction distribution on such a set unless it
contains a hidden Simon pair, an event of probability $O(T^2/N)$.
Consequently, the actual final states in the two cases remain close unless
$T=\Omega(\sqrt N)$, giving the required lower bounds.
  
\end{abstract}
\thispagestyle{empty}
\newpage
\tableofcontents
\newpage

\section{Introduction}
Quantum query complexity measures how many times a quantum algorithm must
access an unknown input to learn a property of it.  The way the input is
made available is part of the model.  The standard quantum oracle for a
function $f$ preserves the input and adds the function value to a separate
register by bitwise XOR, $\ket{x,z}\mapsto\ket{x,z\oplus f(x)}$.  For an
injective function, another natural interface replaces the input by its
image, $\ket x\mapsto\ket{f(x)}$.  Following Aaronson, we call this an
\emph{erasing oracle}.  It preserves quantum coherence because distinct
inputs have distinct outputs.  For a permutation it is a unitary acting on
a single register, called a \emph{minimal oracle} by Kashefi, Kent, Vedral,
and Banaszek~\cite{Kashefi2002} and an \emph{in-place oracle} in later
work~\cite{FeffermanKimmel18,Holman2025}.  Classically, evaluating a
function while retaining the input, or replacing the input by the function
value, provides essentially the same information.  Quantumly, the
distinction can affect which interference patterns an algorithm can create.

Given both an erasing oracle and its inverse, a standard query can be
simulated with two calls (apply the oracle, XOR its output into the target,
and apply the inverse), so the interfaces can differ substantially only when the inverse is
not supplied.  The ability to apply an unknown unitary does not, by itself,
give an algorithm query access to its inverse: for a general unknown
$d$-dimensional unitary, implementing the inverse exactly from forward calls
requires at least $d-1$ calls~\cite{Quintino2019}.  We study these
interfaces when only forward calls are supplied.

Erasing access can be much more powerful than standard access. 
Answering
a question of Kashefi, Aaronson studied the Set Comparison
problem and established a decision separation in favor of erasing
access~\cite{Aaronson2002Collision}.  His later Problem~11 asked for one in
favor of standard access~\cite{Aaronson2021Open}. 
His suggested candidate was
Simon's problem with random garbage appended to each answer.  In Simon's
problem, the function is either one-to-one or has an unknown nonzero
XOR-shift pairing inputs with equal outputs.  With standard queries,
Simon's algorithm detects this hidden structure using only linearly many
queries in the input length~\cite{Simon1997}, and a second standard query
can uncompute the garbage.  Aaronson suggested that, with an erasing
oracle, enough garbage might make the queries no more useful than classical
ones. 
\hide{
For permutations, Holman, Ramachandran, and Yirka separated standard and
forward in-place access for Function Erasure, the subspace-conversion task
of mapping $\sum_{w}\alpha_w\ket{w}\ket{P(w)}$ to $\sum_{w}\alpha_w\ket{w}$
for a permutation $P$ of an $L$-element set: it requires one standard query
but $\Theta(\sqrt L)$ forward in-place queries.  They also proposed a
decision problem as a candidate for a separation in this
direction~\cite{Holman2025}.
}
For permutations, Holman, Ramachandran, and Yirka separated
standard and forward in-place access for a subspace-conversion task and
proposed a decision problem as a candidate for a separation in this
direction~\cite{Holman2025}.

\paragraph{Our results.}
We answer Aaronson's question affirmatively using his proposed
construction.  The garbage is a table of tags: each input has a fixed tag
that is returned on every query to that input.  For $h$ on $n$-bit strings,
promised to be either a permutation or a Simon function, and an arbitrary
tag table $r$, we consider the injection $f_{h,r}(x)=(h(x),x,r_x)$.
Deciding which case holds requires at most $n+2$ standard queries but
$\Theta(2^{n/2})$ forward-erasing queries.  The lower bound matches the
classical collision-sampling algorithm, so for this problem forward-erasing
queries give no asymptotic advantage over classical queries.  We obtain the
same bounds for a promise problem on permutations with forward in-place
access.  This gives a Boolean decision problem separating the two
interfaces for permutations. It is not the candidate proposed by
Holman, Ramachandran, and Yirka.  The lower bounds also hold when a control
qubit coherently selects whether to apply each query.
Section~\ref{sec:overview} states the bounds precisely.

The copy of $x$ in $f_{h,r}(x)$ makes the function injective for every tag
table and simplifies the proofs.  It also means that, for these functions,
a forward-erasing query is the same as a standard query to
$x\mapsto(h(x),r_x)$ whose target register must start in $\ket0$.  The
separation therefore comes from the absence of uncomputation rather than
from the loss of the input.  For permutations, we prove the same bounds
when the copy of $x$ in the prescribed outputs is replaced by a single bit
$b_x$ chosen so that $x\mapsto(h(x),b_x)$ is injective.  On the embedded
instance, the permutation then acts as $x\mapsto(h(x),b_x,r_x)$ up to a
marker bit, which has the form $\langle h(x),\mathrm{gar}(x)\rangle$
suggested by Aaronson, except that injectivity is enforced by the bit $b_x$
rather than holding with high probability over long random garbage.  The
resulting permutation problem has $L=4N^2$ elements, where $N=2^n$, and
requires $O(\log L)$ standard queries versus $\Theta(L^{1/4})$ forward
in-place queries.

\paragraph{Techniques.}
The difficulty is proving a limitation on every forward-erasing algorithm,
rather than showing that Simon's usual algorithm fails.  Appending fixed
garbage preserves the coherence of the full output state.  A quantum
algorithm may keep the garbage, query the same input again, and combine
all of its registers coherently.  Thus one cannot infer a classical lower
bound merely by ignoring the garbage registers.  Moreover, a separation for
simulating one oracle by another need not give a separation for deciding a
property of the underlying function, and Problem~11 asks for the latter.

Our proof builds on the recording method~\cite{Zhandry2019}.  We purify the
random tag table and replace each query by a nearby linear contraction that
records the table cell it reads.  After $T$ queries, averaging over the
tags gives a sum of positive-semidefinite operators, each depending on the
hidden function at no more than $T$ addresses.  A positivity argument
converts this local dependence into a bound on final-state
distinguishability, yielding the same birthday-paradox obstruction as
classical collision sampling: a set of $T$ addresses is unlikely to contain
a hidden Simon pair unless $T=\Omega(\sqrt N)$.  This makes Aaronson's suggested connection
between garbage and decoherence quantitative for our construction.  For
permutations, an algorithm can also query inputs outside the embedded
Simon instance.  We couple the random completion of the permutation to a
uniform reference permutation independent of the hidden instance, and
construct a recording contraction for the resulting comparison query.  The
main additional step is proving that this map is a contraction when
queries inside and outside the embedded instance interfere coherently.

\subsection{Related work}
\label{subsec:related}

\paragraph{Comparing oracle interfaces.}
Kashefi, Kent, Vedral, and Banaszek studied standard and minimal quantum
oracles and the cost of simulating one interface with the
other~\cite{Kashefi2002}, and Atici compared a broader family of quantum
oracle models~\cite{Atici2004}.  Such simulation questions ask an algorithm
to reproduce an oracle's action on arbitrary quantum inputs.  A separation
for this task does not automatically give a separation for deciding a
Boolean property of the underlying function. 
Aaronson considered Set Comparison: given
two injective sequences of length $m$ with entries in $\{1,\ldots,2m\}$,
promised to be either equal as sets or far apart (their union has at least
$1.1m$ elements), decide which.  He showed that $O(1)$ erasing queries
suffice, whereas $\Omega(m^{1/7})$ standard queries are
necessary~\cite{Aaronson2002Collision}. 
The lower bound was raised to 
$\Omega(m^{1/3})$ 
later for the special case of Set Equality, where the sets are
either equal or disjoint~\cite{Zhandry2015,BelovsRosmanis2013}.  Since a query input has $O(\log m)$ bits, this is an
exponential separation in favor of erasing access. 
Aaronson later asked whether the reverse
separation is possible: are there properties of injective functions that
are easy to learn with standard queries but hard with erasing queries
(Problem~11 in~\cite{Aaronson2021Open})?
Our results supply such
a separation, first for injections and then for permutations, with bounded
error on every promised instance.

\paragraph{Forward access, inversion, and search.}
Holman, Ramachandran, and Yirka showed that permutation inversion, which is
as hard as unstructured search~\cite{Nayak2011}, can be solved with a
number of forward in-place queries proportional to the square root of the
domain size~\cite{Holman2025}.  Thus the absence of inverse queries need
not eliminate the usual quadratic search speedup.  They also separated
XOR and in-place access for Function Erasure, a task requiring coherent
conversion of quantum states, and proposed a decision problem as a
candidate for a separation.  Our permutation result concerns a different
Boolean promise problem and does not establish the conjectured bound for
their candidate.  Their work builds on the adversary framework of Belovs
and Yolcu, which accommodates unidirectional access and nonunitary input
oracles~\cite{BelovsYolcu23}.  In a broader unitary-oracle setting, Tang
and Wright showed that generic amplitude amplification and estimation can
lose their quadratic speedups without inverse access~\cite{TangWright25}.
Our result concerns a different setting: an exponential separation between
interfaces to a classical function, even when that function is a
permutation.

\paragraph{Collision finding.}
The Brassard--H{\o}yer--Tapp algorithm~\cite{Brassard1998Collision} finds
a collision in a two-to-one function on $N$ inputs with $O(N^{1/3})$
standard queries, and this is
optimal~\cite{AaronsonShi2004,Ambainis2005,Kutin2005}; see
also~\cite{Zhandry2015} for distinguishing random two-to-one functions from
random permutations.  Since our YES instances contain a two-to-one
function, this may appear to contradict the $\Omega(\sqrt N)$ forward-query
lower bounds.  The distinction is the oracle interface: the collision
algorithm searches for an input whose function value matches a stored list,
using standard queries to compute and uncompute that test inside Grover
search.  Our forward interfaces do not supply this uncomputation.  Also,
the complete oracle functions in our construction are injective, so
collisions are sought only in the hidden function's output coordinate.  For
these instances, our lower bounds rule out other forward-only
implementations of collision finding as well.

\paragraph{Oracles with internal randomness.}
Garbage that an algorithm cannot remove is related to oracles with hidden
randomness.  Regev and Schiff showed that Grover's speedup disappears when
each query fails independently with constant
probability~\cite{RegevSchiff2008}.  Harrow and Rosenbaum studied a
generalization of the standard oracle in which the oracle applies a
permutation selected by hidden internal coins~\cite{HarrowRosenbaum2014}.
In our setting the oracle is a fixed classical function: the garbage is
chosen once and returned identically on every query, so an algorithm may
repeat queries and interfere their outputs.  Randomness enters only through
the hard input distribution used in the lower bound.

\paragraph{Nonstandard oracles in quantum complexity.}
Oracle interfaces also matter when algorithms receive proofs.
Fefferman and Kimmel used in-place permutation oracles to separate
verification with quantum witnesses from verification with classical
witnesses~\cite{FeffermanKimmel18}.  Bassirian, Fefferman, and Marwaha
studied how oracle design changes the complexity of graph property
testing with witnesses~\cite{BFM23}.  More recently, Miloschewsky, Podder,
and Rudolph derandomized the permutation-oracle separation of Fefferman
and Kimmel~\cite{MPR26}.  These results concern verification with
witnesses.  The separation here concerns ordinary quantum decision
algorithms receiving no witness.

\paragraph{Recording queries.}
Our proof combines a query hybrid argument, which replaces oracle calls
one at a time~\cite{Bennett1997}, with a representation of fixed oracle
randomness inspired by Zhandry's query-recording and compressed-oracle
method~\cite{Zhandry2019}.  Hamoudi and Magniez used recording queries to
prove quantum time-space lower bounds for finding multiple collision pairs
in a random function and to give a simpler proof of the quantum sorting
time-space lower bound~\cite{Hamoudi2023Collision}.  Chung, Fehr, Huang,
and Liao extended the compressed-oracle framework to parallel quantum
queries and proved lower bounds for finding chains formed by repeatedly
applying a random function~\cite{Chung2021Compressed}.  Recording and
superposition-oracle techniques have also been developed for random
permutations~\cite{Rosmanis2021,MajenzMalavoltaWalter2025,Carolan2025}.  We
do not need such a technique: we record only cells of the garbage table
attached to a structured function, and we handle the random completion of
the permutation by coupling it to a uniform reference permutation that is
independent of the hidden instance.  We also use an approximate recording
map, rather than an exact compressed simulation.  The local recording
theorem (Theorem~\ref{thm:local-recording}) isolates sufficient conditions
for this argument.

\section{Results and technical overview}
\label{sec:overview}

We first define the oracle interfaces and promise problems, then state the
separations and explain the common lower-bound argument. Standard notation and facts are
collected in Appendix~\ref{sec:preliminaries}.

\subsection{Query models}
\label{subsec:query-models}

\begin{definition}[Standard XOR oracle]
\label{def:standard-xor}
For \(f:\{0,1\}^k\to\{0,1\}^{\ell}\), the standard oracle is the unitary
\(O_f\ket{x,z}=\ket{x,z\oplus f(x)}.\)
\end{definition}

\begin{definition}[Forward-erasing oracle]
\label{def:forward-erasing}
For an injection \(f:\{0,1\}^k\to\{0,1\}^{\ell}\), the forward-erasing
oracle is the isometry
\(E_f\ket x=\ket{f(x)}.\)
Only forward calls to \(E_f\) are supplied: neither \(E_f^\dagger\) nor a
unitary extension of \(E_f\) is available.
\end{definition}

Injectivity ensures that distinct input basis states have orthogonal
outputs, so the map preserves inner products.  We follow Aaronson's
terminology~\cite{Aaronson2021Open}: the term \emph{erasing} describes the
replacement of the input register by the function value, which may itself
contain a copy of the input.

\begin{definition}[Forward in-place oracle]
\label{def:forward-in-place}
For a permutation \(P:\Omega\to\Omega\) of a finite set, the forward
in-place oracle is the unitary
\(U_P\ket w=\ket{P(w)}.\)
Only forward calls to \(U_P\) are supplied, not calls to
\(U_P^\dagger=U_{P^{-1}}\).
\end{definition}

An operation is \emph{oracle-independent} if it does not depend on the
unknown oracle instance.  Between queries an algorithm may perform arbitrary oracle-independent
quantum operations, introduce auxiliary registers, and make adaptive
measurements.  For a Boolean promise problem, let \(Q_{\mathrm{XOR}}\),
\(Q_{\mathrm{erase}}\), and \(Q_{\mathrm{IP}}\) denote the minimum number
of queries in the respective models needed to succeed with probability
at least \(2/3\) on every promised input.

\subsection{Promise problems}
\label{subsec:overview-problems}

Throughout the main text, fix \(n\geq1\), let
\(X=\mathbb F_2^n\), and write \(N=|X|=2^n\).
A \emph{Simon function} \(h:X\to X\) has a nonzero shift \(s\in X\)
such that, for all \(x,y\in X\),
$h(x)=h(y)$ if and only if $y\in\{x,x\oplus s\}$.
Thus every value in its image has exactly two preimages, paired by the
same shift.  For a nonzero \(s\), write
\(s^\perp=\{y\in X:y\cdot s=0\}\), where the inner product is over
\(\mathbb F_2\).
Throughout, \(h\) is promised to be either a permutation or
a Simon function.  A \emph{tag table} is a function \(r:X\to X\), written
\(r=(r_x)_{x\in X}\).  It is part of the oracle instance and remains
fixed throughout the computation.

\begin{definition}[Injective promise problem]
\label{def:garbage-simon}
An instance of \(\mathrm{GarbageSimon}_n\) is the injection
\(f_{h,r}:X\to X^3\) defined by \(f_{h,r}(x)=(h(x),x,r_x)\).
The answer is NO when \(h\) is a permutation and YES when \(h\) is a
Simon function.  Every tag table \(r\) is allowed in both cases.
\end{definition}

To obtain a problem on permutations, we prescribe the same outputs on a
subset of a larger domain and allow every permutation extending that
prescription.  We set \(\Omega=\{0,1\}\times X^3\) and
\(L=|\Omega|=2N^3\), and, for \(x\in X\), define
$a_x=(0,x,0^n,0^n)$ and $F_x=(1,h(x),x,r_x)$.
We write \(\mathcal A=\{a_x:x\in X\}\) and
\(\mathcal F=\{F_x:x\in X\}\).  We call \(\mathcal A\) the set of
\emph{special inputs}.  All other inputs are \emph{background inputs}.
The outputs \(F_x\) are distinct, and their first bit separates them
from \(\mathcal A\).

\begin{definition}[Permutation promise problem]
\label{def:permutation-garbage-simon}
An instance of \(\mathrm{PermutationGarbageSimon}_n\) is a permutation
\(P:\Omega\to\Omega\) satisfying \(P(a_x)=F_x\) for all \(x\in X\),
for a promised \(h\) and a tag table \(r\).  The answer is NO when
\(h\) is a permutation and YES when \(h\) is a Simon function.
Every tag table and every bijection from
\(\Omega\setminus\mathcal A\) to \(\Omega\setminus\mathcal F\) are
allowed.
\end{definition}

Such a bijection is called a \emph{completion}.  Every pair \((h,r)\)
admits one.  An algorithm may query special and background inputs in
superposition, so a lower bound must account for the information in the
completion as well as in the prescribed outputs.

\subsection{Main results}
\label{subsec:main-results}

\begin{theorem}[Combined separation]
\label{thm:combined-separation}
For every \(n\geq1\),
\[
 \begin{aligned}
 Q_{\mathrm{XOR}}(\mathrm{GarbageSimon}_n)&\leq n+2,
 &Q_{\mathrm{erase}}(\mathrm{GarbageSimon}_n)&=\Theta(2^{n/2}),\\
 Q_{\mathrm{XOR}}(\mathrm{PermutationGarbageSimon}_n)&\leq n+2,
 &Q_{\mathrm{IP}}(\mathrm{PermutationGarbageSimon}_n)&=\Theta(2^{n/2}).
 \end{aligned}
\]
\end{theorem}

The injective statement combines Propositions~\ref{prop:xor-upper}
and~\ref{prop:erase-upper} with
Corollary~\ref{cor:injective-from-recording}.  The permutation upper
bounds are proved in Section~\ref{sec:permutation-extension}, and the
permutation lower bound is
Theorem~\ref{thm:permutation-lower-bound-estimate}.  Since \(N=2^n\), both
statements give a separation of $O(n)$ standard queries versus $\Theta(2^{n/2})$ forward queries.  The
permutation construction described above has \(L=2N^3\) elements, so its
forward complexity is \(\Theta(L^{1/6})\).

The explicit copy of \(x\) simplifies the proofs by making the relevant
address visible in each prescribed output.  Replacing it by a single bit
gives a compact permutation problem on \(L=4N^2\) elements with the same
bounds, so
the forward in-place complexity improves to \(\Theta(L^{1/4})\), while
\(n+2=O(\log L)\) standard queries still suffice
(Section~\ref{subsec:compact}).

The lower bounds also hold if the algorithm can coherently choose between
applying a query and an oracle-independent bypass.  For an erasing oracle,
the bypass is a fixed isometry into the output space.  For an in-place
oracle, it can be the identity.  Proposition~\ref{prop:coherent-bypass} gives
the formal statement and proof.

\subsection{Why the upper bounds hold}
\label{subsec:overview-upper}

A standard query can isolate the \(h\)-coordinate.  We initialize the
$h$-value target register to $\ket{0^n}$ and each qubit in the other
target registers to
\(\ket+=(\ket0+\ket1)/\sqrt2\).  XOR leaves \(\ket+\) unchanged,
so the address and tag fields do not become entangled with the query.
A uniform superposition of inputs therefore yields the usual Simon state
$N^{-1/2}\sum_{x\in X}\ket x\ket{h(x)}$.
For the permutation problem, we use the uniform superposition of special
inputs and initialize the marker target to $\ket+$ as well.
To obtain a \emph{Simon sample} (see Appendix~\ref{subsec:simon-sampling}), we measure the $h$-value register, apply
the $n$-qubit Hadamard transform to the input register, and measure that
register.
The resulting sample is uniform in \(X\) when \(h\) is a permutation,
and uniform in \(s^\perp\) when \(h\) is a Simon function with shift $s$.
With \(n+2\) samples, testing whether these vectors span
\(X\) distinguishes the two cases with the required success probability.

For either forward interface, a classical strategy suffices for the
matching upper bound.  We query a uniformly random set of
\(\Theta(\sqrt N)\) distinct addresses (special inputs in the permutation
case) and inspect the returned \(h\)-values.  There is no collision in
the permutation case.  In the Simon case the sample contains a paired
set of addresses with constant probability, by the birthday paradox: a
uniformly random set of $k$ distinct addresses contains a Simon pair with
probability at least $1-\exp(-k(k-1)/(2N))$
(Proposition~\ref{prop:erase-upper}).

\subsection{Proof strategy}
\label{subsec:proof-strategy}
\label{subsec:overview-recording}

The challenge is to show that coherent forward queries cannot improve
on collision sampling.  For the injective lower bound, we choose \(h\)
uniformly from permutations in the NO case.  In the YES case, we choose a
uniform nonzero shift and assign distinct, uniformly random output labels
to its pairs.  In both cases, we choose each tag independently and uniformly
from \(X\), once for the entire oracle instance.  For the permutation
lower bound, we also choose a uniform completion, fixed for the computation.
An algorithm succeeding on every promised input must distinguish these
two distributions.  As explained in the introduction, discarding the tags
in the analysis does not give a lower bound.  We instead compare the entire
computation with an auxiliary one, used only in the proof, whose dependence
on \(h\) can be controlled.

\paragraph{Recording the fixed table.}
We represent each random tag by an inaccessible register \(R_x\), initially
in the uniform superposition
\(\ket u=N^{-1/2}\sum_{z\in X}\ket z\).
Queries read these registers without changing their computational-basis
labels, so tracing them out reproduces an average over one fixed classical
table, rather than over fresh randomness at each call.  We say that a cell
is \emph{unrecorded} if its state lies in the span of \(\ket u\), and
\emph{recorded} if it lies in the orthogonal complement.

We replace each query by a \emph{contraction}, a linear map of operator
norm at most one, that is close to the query in operator norm and has
three properties.  It records at most one new cell, never removes a
record, and ensures that the component whose recorded set is
\(D\subseteq X\) depends on \(h\) only through its restriction
\(h|_D:D\to X\).  After \(T\) queries, all such sets have size at most
\(T\).  Components with different recorded sets occupy orthogonal
subspaces of the table, which we call \emph{sectors}.  Tracing out the
table removes cross terms between sectors, leaving a sum of
positive-semidefinite operators, each depending only on the hidden
function at a small set of addresses.

\paragraph{From local dependence to a collision bound.}
For any fixed set \(D\), the restriction of a random Simon function has
the same distribution as the restriction of a random permutation unless
\(D\) contains both members of a Simon pair.  For a nonzero shift chosen
uniformly at random, the probability of this exception is at most
\(\binom{|D|}{2}/(N-1)\).  The weight of each sector contribution may
depend on \(h\), so the recorded sets cannot be treated as a classical
sample chosen independently of the hidden function.  Positivity lets us
weight the exceptional probability by the trace of each contribution, which
bounds the distinguishability of the final states even for an adaptive
quantum strategy.

More precisely, suppose each query is within operator-norm distance
\(\delta\) of a contraction satisfying the recording properties and that
$T\delta\leq1$.  Then the local recording theorem
(Theorem~\ref{thm:local-recording}) gives
$\|\rho_0-\rho_1\|_1\leq T(T-1)/(N-1)+6T\delta$.
Here \(\rho_0,\rho_1\) are the final states averaged over the two hard
distributions, and \(\|\cdot\|_1\) is the trace norm.  By Helstrom's
bound, their optimal discrimination probability with equal prior
probabilities is $\tfrac12+\tfrac14\|\rho_0-\rho_1\|_1$
\cite[Theorem~3.4]{Watrous2018}.  Thus a per-query error of
\(O(N^{-1/2})\) yields an \(\Omega(\sqrt N)\) query lower bound.  The
theorem uses the hard distributions only through the comparison of
restrictions described above, which lets the same theorem cover the
compact permutation problem.

\paragraph{The injective application.}
\label{subsec:overview-injective}
A forward-erasing query returns \((h(x),x,r_x)\).  Projecting the
corresponding table cell \(R_x\) onto its recorded subspace gives the
required contraction.  Copying a uniform tag into the output leaves
only a \(1/N\) fraction of the squared norm in the unrecorded
subspace.  The returned address makes the errors for distinct addresses
orthogonal, giving operator-norm error \(N^{-1/2}\), and
Theorem~\ref{thm:local-recording} then proves the injective lower bound.
This is where the fixed random tags enter the argument: a constant tag
table would still permit Simon sampling with forward access.

\paragraph{The permutation application.}
\label{subsec:overview-permutation}
In the permutation problem, the prescribed outputs \(F_x\) on the special
inputs \(a_x\) form an embedded copy of the injective instance, but an
in-place algorithm may also query the background inputs, whose images are
given by the random completion.  This raises two issues.

First, the completion depends on the hidden instance, because it must avoid
the prescribed outputs.  We show that a uniform completion can be written
as a uniform permutation \(\pi\) of \(\Omega\), independent of \((h,r)\),
composed with a permutation that moves only a random set of at most \(2N\)
elements (Lemma~\ref{lem:reverse-coupling}).  A hybrid argument then
replaces the real oracle by a comparison oracle built from \(\pi\) and the
prescribed outputs, at a cost of \(O(T\sqrt{N/L})\) in trace norm.  In the
comparison oracle, a background query depends on the hidden instance only
when its image under \(\pi\) has the form of a prescribed output \(F_i\),
which amounts to testing whether \(h(i)\) and the tag \(r_i\) take given
values.

Second, we construct a recording contraction for the comparison oracle.
For such a background query, we keep the tag-independent part of its action
on the unrecorded part of cell \(R_i\) and remove only the small part that
depends on the tag.  Unlike in the injective case, the unrecorded part
cannot simply be discarded, because it carries most of the norm.
Moreover, a query to the special input \(a_i\) and the background queries
whose images involve the index \(i\) have overlapping outputs and can
interfere, so the contraction property must be checked for their coherent
combinations rather than branch by branch.  A direct calculation shows that
the norm removed by the recording projections dominates all interference
terms (Lemma~\ref{lem:permutation-one-query-recording}), giving a
per-query error of at most \(4/\sqrt N\).

Together with the completion hybrid, Theorem~\ref{thm:local-recording}
bounds the trace distance between the real final states in the two cases by
$T(T-1)/(N-1)+24T/\sqrt N+16T\sqrt{N/L}$ when $4T\leq\sqrt N$
(Theorem~\ref{thm:permutation-lower-bound-estimate}), which implies the
permutation lower bound.  The same estimates apply to the compact
construction with $L'=4N^2$.

\paragraph{Organization.}
Section~\ref{sec:recording-theorem} states and proves the local recording
theorem, the lower-bound argument shared by both oracle models.
Section~\ref{sec:injective-lower-bound} applies it to the forward-erasing
query for $\mathrm{GarbageSimon}_n$, which proves the injective lower bound
and answers Problem~11.  Section~\ref{sec:permutation-extension} proves the
lower bound for $\mathrm{PermutationGarbageSimon}_n$ extending the result in Section~\ref{sec:injective-lower-bound}.
Section~\ref{sec:combined-discussion} discusses open problems.
Appendix~\ref{sec:preliminaries} collects notation and standard facts and
proves that the lower bounds extend to coherently bypassable queries.
The remaining appendices contain the proofs and calculations deferred from
the main text.

\section{Local recording theorem}
\label{sec:recording-theorem}

This section isolates the lower-bound argument shared by the two oracle
models.

By purification and deferred measurement~\cite{Watrous2018}, we may assume
that the algorithm starts in a pure state and uses only oracle-independent
isometries between queries.
The system $\mathsf{Alg}$ includes its query registers, workspace, and
purifying registers.  

Let $\Hcal_0$ and $\Hcal_1$ be the following distributions on functions
$h:X\to X$.
Under $\Hcal_0$, the map $h$ is a uniformly random permutation of $X$.
Under $\Hcal_1$, we first choose a shift $s$ uniformly from
$X\setminus\{0\}$.  Let
\(\mathcal P_s=\{\{x,x\oplus s\}\mid x\in X\}\) be the resulting set of $N/2$ disjoint unordered pairs.  We choose a uniformly
random injection $\lambda:\mathcal P_s\hookrightarrow X$ and set
$h(x)=\lambda(\{x,x\oplus s\})$.  Thus a sample from $\Hcal_1$ satisfies
\(h(x)=h(y)\) if and only if \(y\in\{x,x\oplus s\}\).

The lower-bound argument uses these distributions only through their
restrictions to small sets of addresses.  To cover the compact permutation problem
of Section~\ref{subsec:compact} as well, we state it for more general hidden
data: a function $h:X\to\Sigma$ for a finite set $\Sigma$, together with
two distributions $\Hcal_0,\Hcal_1$ on such functions.
Throughout the recording theorem and its sector induction, the hidden
functions $h,h'$ range over
$\supp(\Hcal_0)\cup\supp(\Hcal_1)$; the query-map hypotheses and locality
comparisons are required only for these functions.
For $D\subseteq X$
and $i\in\{0,1\}$, let $\mu_{i,D}$ be the distribution of the restriction
$h|_D$ when $h\sim\Hcal_i$.  We require the following \emph{restriction
property}: for every $D\subseteq X$ there are a number $p_D$ and a
probability distribution $\nu_D$ on functions $D\to\Sigma$ such that
\begin{equation}
   \mu_{1,D}=(1-p_D)\mu_{0,D}+p_D\nu_D,
   \qquad
   0\leq p_D\leq\frac{\binom{|D|}{2}}{N-1}.
   \tag{R}\label{eq:restriction-property}
\end{equation}

\begin{lemma}[Simon restrictions]
\label{lem:simon-restrictions}
The distributions $\Hcal_0$ and $\Hcal_1$ defined above, with $\Sigma=X$,
satisfy~\eqref{eq:restriction-property} with $p_D=|\Delta(D)|/(N-1)$,
where $\Delta(D)=\{x\oplus y\mid x,y\in D,\ x\neq y\}$.
\end{lemma}

\begin{proof}
Fix $D$.  Under $\Hcal_0$, the restriction $h|_D$ is a uniformly random
injection $D\hookrightarrow X$.  Under $\Hcal_1$, condition on $s$.  If
$s\notin\Delta(D)$, no pair in $\mathcal P_s$ contains two distinct points
of $D$.  Hence $x\mapsto\{x,x\oplus s\}$ is injective on $D$, and the
random injection $\lambda$ again makes $h|_D$ a uniformly random injection
$D\hookrightarrow X$.  Since $s$ is uniform on $X\setminus\{0\}$ and
$0\notin\Delta(D)$, the event $s\in\Delta(D)$ has probability
$p_D=|\Delta(D)|/(N-1)\leq\binom{|D|}{2}/(N-1)$.  Taking $\nu_D$ to be the
conditional distribution of $h|_D$ given $s\in\Delta(D)$ proves the claim.
If $p_D=0$, the choice of $\nu_D$ is immaterial.
\end{proof}

Independently of $h$, we sample the tag table $r$ uniformly from the set $X^X$
of all functions $X\to X$ and keep it fixed throughout the computation.
To represent the resulting classical mixture by a single pure
global evolution, we introduce one
inaccessible table-cell register for each address.  For $i\in X$, let
$R_i\cong\mathbb C^N$ have computational basis
$\{\ket z\mid z\in X\}$.  Let $\ket u=N^{-1/2}\sum_{z\in X}\ket z$ and define
\(\mathcal R=\bigotimes_{i\in X}R_i\) and
\(\ket U=\bigotimes_{i\in X}\ket u\).
Expanding $\ket U$ in the computational basis gives
$\ket U=N^{-N/2}\sum_{r\in X^X}\ket r$, where $\ket r=\bigotimes_{i\in X}\ket{r_i}$.
On a branch labelled by a basis table $\ket r$, the coherent extension of a
fixed-table query behaves as the classical oracle with table $r$ and leaves the label
$\ket r$ unchanged.  It may nevertheless use that label as a coherent
control and entangle $\mathcal R$ with the algorithm system.  Tracing out
$\mathcal R$ after all queries removes cross terms between distinct basis
tables and gives exactly the classical average over one uniformly random
table sampled at the beginning and reused on every query.  In particular,
repeated queries to address $i$ use the same cell $R_i$ and the same value
$r_i$.

We write $\operatorname{im}A$ for the range of a linear map $A$.
On one cell we define \(\Pi_{0,i}=(\ket u\!\bra u)_{R_i}\) and
\(\Pi_{1,i}=\Id-\Pi_{0,i}\).
We call $\operatorname{im}\Pi_{0,i}$ the \emph{unrecorded subspace} and
$\operatorname{im}\Pi_{1,i}$ the \emph{recorded subspace} of $R_i$.
For $D\subseteq X$, define the {\em table sector}
\[\mathcal R_D=\bigotimes_{i\in D}\operatorname{im}\Pi_{1,i}\otimes\bigotimes_{i\notin D}\operatorname{im}\Pi_{0,i},\]
and let $\Pi_D$ denote the projector onto $\mathcal R_D$.
The sectors give the orthogonal decomposition
$\mathcal R=\bigoplus_{D\subseteq X}\mathcal R_D$, and $\ket U$ lies in the
empty sector $\mathcal R_\varnothing$.

We now specify the query isometries and their approximations for the
local recording theorem.  A \emph{query slot} is one position
$j\in\{1,\ldots,T\}$ in the sequence of oracle calls, together with
any fixed auxiliary-register initialization used to represent that call.
The oracle-independent operations between successive slots are specified
by the algorithm.  A \emph{retained system} consists of all registers
other than the inaccessible table register $\mathcal R$.  It includes
$\mathsf{Alg}$ and any auxiliary registers introduced by the analysis.
The algorithm's later operations act as the identity on auxiliary
registers introduced only by the analysis.  We write $\mathsf S$ for
the retained system when its position in the computation is clear.

Let $\theta$ be an auxiliary random parameter, independent of $h$ and
$r$, with the same distribution in the two promise classes.  In the
injective application it is deterministic.  In the permutation
application it is the independent reference permutation $\pi$ defined
in Lemma~\ref{lem:reverse-coupling}.

For slot $j$, write $\mathsf S_{\rm in}$ and $\mathsf S_{\rm out}$ for
the retained systems immediately before and after the $j$th query.  For fixed
$h$ and $\theta$, let the query being approximated be an isometry
$\mathcal Q_{h,\theta}:\mathsf S_{\rm in}\otimes\mathcal R\to\mathsf S_{\rm out}\otimes\mathcal R$.
For each basis table $r\in X^X$, assume there is a fixed-table query
isometry $\mathcal Q^{(r)}_{h,\theta}$ satisfying
$\mathcal Q_{h,\theta}(\ket\psi\ket r)=(\mathcal Q^{(r)}_{h,\theta}\ket\psi)\ket r$ for $\ket\psi\in\mathsf S_{\rm in}$.
Thus the query may use $r$ coherently but leaves its basis label unchanged.
We compare this isometry with a linear map
$\widetilde{\mathcal Q}_{h,\theta}$ having the same domain and codomain.
If the latter introduces an auxiliary output register, we append that
register in a fixed zero state to the original query as well.  This
convention lets us compare the two maps in operator norm.  Registers on
which a map acts as the identity are called \emph{spectator
registers} and are omitted from its formula.  We suppress the index $j$
on the query maps, but require the following conditions at every slot
with a common error bound $\delta$.

The query map need not respect the sector decomposition: it may place
$h(i)$ in an algorithm-output register while cell $R_i$ still has an
unrecorded component.  The proof therefore compares it with a linear map
$\widetilde{\mathcal Q}_{h,\theta}$ that is close enough to control the hybrid
error while satisfying the sector-transition and locality conditions below.
Here \emph{locality} means dependence on $h$ only at addresses contained
in the output sector's recorded set.
Throughout the recording arguments, a tilde marks a map, vector, or operator
obtained from this analysis-only replacement.
We call such a map a
\emph{recording contraction} when it satisfies the following three
conditions.

\begin{enumerate}
\item\label{item:recording-close}
For some $\delta\geq0$ independent of $h$, $\theta$, and the query slot,
\(\|\mathcal Q_{h,\theta}-\widetilde{\mathcal Q}_{h,\theta}\|\leq\delta\)
and \(\|\widetilde{\mathcal Q}_{h,\theta}\|\leq1\).

\item\label{item:recording-sector}
For every pair of sectors $D,D'\subseteq X$, define the sector
block by
\[
   \widetilde{\mathcal Q}_{h,\theta}^{D'\leftarrow D}
    = (\Id_{\mathsf S_{\rm out}}\otimes\Pi_{D'})
       \widetilde{\mathcal Q}_{h,\theta}
      (\Id_{\mathsf S_{\rm in}}\otimes\Pi_D).
\]
We require this block to vanish unless \(D\subseteq D'\) and
\(|D'\setminus D|\leq1\).
Thus a query never removes a record and adds at most one new address.
It may leave the recorded set unchanged, including when it revisits an
already recorded address.

\item\label{item:recording-local}
Each sector block ending in $D'$ uses no value of $h$ outside $D'$.
Formally, whenever $h|_{D'}=h'|_{D'}$,
$\widetilde{\mathcal Q}_{h,\theta}^{D'\leftarrow D}=\widetilde{\mathcal Q}_{h',\theta}^{D'\leftarrow D}$.
\end{enumerate}

\Needspace{14\baselineskip}
\begin{theorem}[Local recording theorem]
\label{thm:local-recording}
Let $X=\F_2^n$ and $N=2^n$.  Let $\Hcal_0,\Hcal_1$ be distributions on
functions $h:X\to\Sigma$ satisfying the restriction
property~\eqref{eq:restriction-property}, and let the table state $\ket U$
be defined as above.  Consider a fixed $T$-query algorithm
whose initial state and inter-query operations are independent of $h$ and
$r$.  Let $\theta$ be independent of $h$ and $r$, with the same
distribution in both promise classes.  At every query slot, suppose the
fixed-table query isometries $\mathcal Q^{(r)}_{h,\theta}$ have coherent
extensions $\mathcal Q_{h,\theta}$ preserving the basis table label as
above, and admit recording contractions
$\widetilde{\mathcal Q}_{h,\theta}$ satisfying
conditions~\ref{item:recording-close}--\ref{item:recording-local}
with the same $\delta$.
The table register $\mathcal R$ is initialized in $\ket U$.  For fixed $h$ and $\theta$, let
$\rho_{h,\theta}$ be the density operator on the retained system after
running the purified computation with $\mathcal Q_{h,\theta}$ and
tracing out $\mathcal R$.
Equivalently, $\rho_{h,\theta}=\mathbb E_r[\rho_{h,\theta}^{(r)}]$,
where $r$ is one uniformly sampled table held fixed throughout the
computation and $\rho_{h,\theta}^{(r)}$ is the retained final state for that
fixed table.
For $i\in\{0,1\}$, define
\(\rho_i=\mathbb E_\theta\,\mathbb E_{h\sim\Hcal_i}[\rho_{h,\theta}]\).
Let $\epsilon=T\delta$.
If $\epsilon\leq 1$, then
\[
    \|\rho_0-\rho_1\|_1
       \leq
       \frac{T(T-1)}{N-1}+6\epsilon .
\]
\end{theorem}

\begin{proof}
Fix $\theta$, and for $i\in\{0,1\}$ let
$\rho_{i,\theta}=\mathbb E_{h\sim\Hcal_i}\rho_{h,\theta}$.  Let
$\ket{\Psi_{h,\theta}}$ be the final global
pure state obtained with the query maps $\mathcal Q_{h,\theta}$, and let
$\ket{\widetilde\Psi_{h,\theta}}$ be obtained by replacing every query by
$\widetilde{\mathcal Q}_{h,\theta}$.  A telescoping hybrid over the $T$ query positions
gives
\begin{equation}
   \bigl\|\ket{\Psi_{h,\theta}}
            -\ket{\widetilde\Psi_{h,\theta}}\bigr\|
      \leq \epsilon .
   \label{eq:recording-vector-hybrid}
\end{equation}
Indeed, the isometry assumption on $\mathcal Q_{h,\theta}$ and the
contraction bound in condition~\ref{item:recording-close} ensure that
all maps before and after the changed query have norm at most one.
The same condition bounds the difference at that query by $\delta$.
Summing over the $T$ positions gives $T\delta=\epsilon$.  The
query computation is normalized, while the recording computation is a
composition of contractions.  Hence
\begin{equation}
   (1-\epsilon)^2
     \leq \|\widetilde\Psi_{h,\theta}\|^2
     \leq 1.
   \label{eq:recording-norm-window}
\end{equation}

For each $D\subseteq X$, we first define the sector vector
$\ket{\widetilde\Psi_{h,\theta,D}}=(\Id\otimes\Pi_D)\ket{\widetilde\Psi_{h,\theta}}$.
The simultaneous induction in Appendix~\ref{app:sector-induction} proves
that, for every fixed $\theta$,
$\ket{\widetilde\Psi_{h,\theta}}=\sum_{D\subseteq X,\,|D|\leq T}\ket{\widetilde\Psi_{h,\theta,D}}$.
Moreover, if $h|_D=h'|_D$, then
\(\ket{\widetilde\Psi_{h,\theta,D}}=\ket{\widetilde\Psi_{h',\theta,D}}\).

After tracing out the table register, the cross terms between distinct sectors
vanish.  Fix $\theta$ and suppress it from the notation.  If
$\varphi:D\to \Sigma$ occurs as $h|_D$ for some $h$ in the support of either hard
distribution, define
\begin{equation}
   \sigma(D,\varphi)
      =\Tr_{\mathcal R}
        \left(
          \ket{\widetilde\Psi_{h,\theta,D}}
          \!\bra{\widetilde\Psi_{h,\theta,D}}
        \right),
   \qquad\text{where }h|_D=\varphi.
   \label{eq:sigma-D-varphi-definition}
\end{equation}
This is well defined by locality.  For a function $\varphi$ that is not an
attainable restriction in either hard distribution, set
$\sigma(D,\varphi)=0$.  Such values never enter the averages below.  For every
pair $(D,\varphi)$, the value $\sigma(D,\varphi)$ is a positive semidefinite
operator on the final retained system $\mathsf S$.  The recording computation
has retained operator
$\widetilde\rho_{h,\theta}=\sum_{D\subseteq X,\,|D|\leq T}\sigma(D,h|_D)$.

By the restriction property~\eqref{eq:restriction-property} (verified for the original hard distributions in Lemma~\ref{lem:simon-restrictions} and for the compact variant in Section~\ref{subsec:compact}), for every
$D$ the two restriction distributions satisfy
$\mu_{1,D}=(1-p_D)\mu_{0,D}+p_D\nu_D$ with $p_D\leq\binom{|D|}{2}/(N-1)$.

Summing a separate error bound for each sector would introduce a factor
equal to the number of sectors.  Instead, we weight each contribution by
its trace.  For fixed $h$, the \emph{weight} of sector $D$ is
$\Tr\sigma(D,h|_D)=\|\widetilde\Psi_{h,\theta,D}\|^2$.  These weights can
depend on $h$, so we cannot treat $D$ as a classical sample independent of
the hidden function.  Positivity and the near-unit total trace of the
recording operators give the required weighted bound.

For the remainder of the proof, every sum over $D$ ranges over subsets of
$X$ of size at most $T$.  We define the positive operators
\begin{equation}
\begin{aligned}
   A_D&=\mathbb E_{\varphi\sim\mu_{0,D}}\sigma(D,\varphi),
   &
   C_D&=\mathbb E_{\varphi\sim\nu_D}\sigma(D,\varphi),
   \\
   B_D&=\mathbb E_{\varphi\sim\mu_{1,D}}\sigma(D,\varphi)
       =(1-p_D)A_D+p_D C_D.
\end{aligned}
\label{eq:recording-ABC}
\end{equation}
Thus $A_D$ is the average contribution of sector $D$ under
$\mu_{0,D}$.  The operator $A_D$ need
not have trace one: its trace is the expected squared norm of the
$D$-sector vector under $\Hcal_0$, for fixed $\theta$.

Let \(a=\Tr\sum_D A_D\) and \(b=\Tr\sum_D B_D\).
Equation~\eqref{eq:recording-norm-window} implies
\(|a-b|\leq1-(1-\epsilon)^2\leq2\epsilon\).
We define \(S_A=\sum_D p_D\Tr A_D\) and
\(S_C=\sum_D p_D\Tr C_D\).
Taking traces in~\eqref{eq:recording-ABC} and summing over $D$ gives
$b=a-S_A+S_C$, and therefore \(S_C\leq S_A+|a-b|\).
Since $A_D,C_D\succeq0$, positivity and the triangle inequality yield
\begin{equation}
\begin{aligned}
   \left\|\sum_D A_D-\sum_D B_D\right\|_1
     &\leq \sum_D p_D\|A_D-C_D\|_1 \\
     &\leq S_A+S_C \\
     &\leq 2S_A+|a-b| \\
     &\leq 2\max_{|D|\leq T}p_D+2\epsilon \\
     &\leq \frac{T(T-1)}{N-1}+2\epsilon.
\end{aligned}
\label{eq:recording-state-bound}
\end{equation}
In the penultimate line we used
$S_A\leq(\max_{|D|\leq T}p_D)a\leq\max_{|D|\leq T}p_D$.

It remains to return from the recording computation to the
$\mathcal Q$-computation.  For a unit vector $\ket\psi$ and a vector
$\ket\phi$ of norm at most one,
$\|\ket\psi\!\bra\psi-\ket\phi\!\bra\phi\|_1\leq 2\|\ket\psi-\ket\phi\|$
(see~\eqref{eq:pure-to-trace}).  Together with
\eqref{eq:recording-vector-hybrid}, partial trace, and averaging, this
bounds the error by $2\epsilon$ in each hard class.
Adding the two hybrid errors to
\eqref{eq:recording-state-bound} proves, for fixed $\theta$,
\(\|\rho_{0,\theta}-\rho_{1,\theta}\|_1\leq T(T-1)/(N-1)+6\epsilon\).
Finally, we average over $\theta$.  Its distribution is identical in the two
hard classes, and convexity of the trace norm preserves the same bound.
\end{proof}

\section{Injective lower bound}
\label{sec:injective-lower-bound}

We verify the hypotheses of Theorem~\ref{thm:local-recording} for the
purified forward-erasing query, which yields the lower bound in
Corollary~\ref{cor:injective-from-recording} below.

The injective oracle is $f_{h,r}(x)=(h(x),x,r_x)$.  Let $V_h$ be the purified forward-erasing query,
acting coherently on the table register $\mathcal R$.  The output
registers $H$, $\mathsf X$, and $Z$ store $h(x)$, $x$, and $r_x$,
respectively.  With input order
$(\mathsf X,\mathcal R)$ and output order $(H,\mathsf X,Z,\mathcal R)$, its basis action is
$V_h\ket x\ket r=\ket{h(x),x,r_x}\ket r$.
The table label is preserved basis-wise but can become entangled with the
algorithm-output tag register $Z$ on a superposition of tables.

The output of a query to $x$ decomposes into recorded and unrecorded
components of $R_x$.  The unrecorded component violates locality: the
output depends on $h(x)$ even in a sector whose recorded set does not contain $x$.
The recording contraction removes this component.

We use the returned address in $\mathsf X$ to define the projector
$\Pi_{\mathrm{rec}}=\sum_{x\in X}(\ket x\!\bra x)_{\mathsf X}\otimes\Pi_{1,x}$,
where identities on $H$, $Z$, the other table cells, and the algorithm
workspace are implicit.  We define \(\widetilde V_h=\Pi_{\mathrm{rec}}V_h\).
In the notation of
Theorem~\ref{thm:local-recording},
$\mathcal Q_{h,\theta}=V_h$,
$\widetilde{\mathcal Q}_{h,\theta}=\widetilde V_h$, and $\theta$ is
deterministic.

To compute the approximation error, fix $i\in X$ and write
$\Pi_0=(\ket u\!\bra u)_{R_i}$ and $\Pi_1=\Id-\Pi_0$.  For an arbitrary
reference system $\mathsf K$, a vector on $R_i\otimes\mathsf K$ has the form
$\sum_z\ket z_{R_i}\ket{\alpha_z}_{\mathsf K}$.  With output order
$(R_i,Z,\mathsf K)$, copying the tag value to $Z$ and projecting the table
cell by $\Pi_0$ gives
$N^{-1/2}\ket u_{R_i}\sum_z\ket z_Z\ket{\alpha_z}_{\mathsf K}$.
Its squared norm is
$N^{-1}\sum_z\|\alpha_z\|^2$, exactly $1/N$ times the input squared norm.
Hence the discarded $\Pi_0$ component has operator norm
$N^{-1/2}$.  The algorithm-output copy of $x$ makes the error ranges for different
addresses orthogonal, so the same norm holds for a coherent superposition of
addresses.  Therefore $\|V_h-\widetilde V_h\|=N^{-1/2}$ and
$\|\widetilde V_h\|\leq1$.
On a basis query branch $x$, the projector $\Pi_{\mathrm{rec}}$ forces cell
$R_x$ into its recorded subspace and acts as the identity on every other
cell.  It therefore maps a table sector $\mathcal R_D$ into
$\mathcal R_{D\cup\{x\}}$.  No recorded cell becomes unrecorded.  The query's
dependence on $h$ is only through $h(x)$, and the output sector contains $x$.
Consequently, the sector-transition and locality conditions
\ref{item:recording-sector} and~\ref{item:recording-local} hold.  Together with the norm
estimate, these observations verify all three conditions with
$\delta=N^{-1/2}$.

\begin{corollary}[Injective trace-distance bound and separation]
\label{cor:injective-from-recording}
Let $\rho^{\mathrm{inj}}_0$ and $\rho^{\mathrm{inj}}_1$ be the average
final states on the algorithm system $\mathsf{Alg}$ of any $T$-query
forward-erasing strategy under
$h\sim\Hcal_0$ and $h\sim\Hcal_1$, respectively, with one uniformly random
tag table sampled at the beginning and reused on every query.  If
$T\leq\sqrt N$, then
\begin{equation}
   \|\rho^{\mathrm{inj}}_0-\rho^{\mathrm{inj}}_1\|_1
      \leq
      \frac{T(T-1)}{N-1}+\frac{6T}{\sqrt N}.
   \label{eq:injective-quantitative-bound}
\end{equation}
Every bounded-error forward-erasing algorithm for
$\mathrm{GarbageSimon}_n$ therefore uses
more than $\sqrt N/16$ queries.
\end{corollary}

\begin{proof}
The one-cell calculation above and the sector transition rule verify the
hypotheses of Theorem~\ref{thm:local-recording} with
$\delta=N^{-1/2}$, and Lemma~\ref{lem:simon-restrictions} verifies the
restriction property~\eqref{eq:restriction-property}.
Since $\epsilon=T/\sqrt N$, that theorem gives
\eqref{eq:injective-quantitative-bound}.

Suppose a strategy succeeds with probability at least $2/3$ on every
promised injection.  It then distinguishes the two equal-prior hard
distributions with success probability at least $2/3$.  If
$T\leq\sqrt N/16$, then, for every $N\geq2$,
\[
 \|\rho^{\mathrm{inj}}_0-\rho^{\mathrm{inj}}_1\|_1
 \leq \frac6{16}+\frac{N}{256(N-1)}
 \leq \frac38+\frac1{128}
 =\frac{49}{128}<\frac23.
\]
Here we used $T(T-1)\leq T^2$ and $N/(N-1)\leq2$.
Helstrom's bound then limits the average success probability to less than
$\frac12+\frac14\cdot\frac23=\frac23$, a contradiction.  Hence
$T>\sqrt N/16$.  Proposition~\ref{prop:erase-upper} supplies the matching
$O(\sqrt N)$ upper bound, and Proposition~\ref{prop:xor-upper} supplies the
standard-query bound.
\end{proof}

Corollary~\ref{cor:injective-from-recording} settles Aaronson's
Problem~11.  The permutation separation proved next strengthens this result.
The permutation proof uses Theorem~\ref{thm:local-recording} itself,
not merely Corollary~\ref{cor:injective-from-recording}.  A permutation
algorithm can query background inputs that do not occur in the injective
interface, so the recording hypotheses must also be verified for those
queries.

\section{Permutation problem}
\label{sec:permutation-extension}

The injective lower bound alone does not imply a permutation lower bound:
queries to background inputs may reveal information about the prescribed
image set.  We first couple a uniform completion to a uniform reference
permutation independent of the hidden instance.  A sparse-disagreement
hybrid reduces the real query process to an ideal comparison map.  We
then construct a recording contraction for that map and apply
Theorem~\ref{thm:local-recording}.  The main additional difficulty is
proving contraction in the presence of interference between special and
background inputs.

Recall Definition~\ref{def:permutation-garbage-simon}: the domain is
$\Omega=\{0,1\}\times X^3$, with $L=2N^3$ elements, and the promise fixes
$P(a_x)=F_x$ for $a_x=(0,x,0^n,0^n)$ and
$F_x=(1,h(x),x,r_x)$.  We retain the notation $\mathcal A$ for the special
inputs and $\mathcal F$ for their prescribed images.  We write elements of
$\Omega$ as $(\beta,u,v,t)$, with $\beta$ a bit and $u,v,t\in X$.
Under the identification $\Omega\cong\{0,1\}^{3n+1}$, bitwise XOR acts
componentwise on these tuples.

\begin{remark}[Input copy]
\label{rem:input-copy}
The explicit $x$-coordinate makes the values $F_x$ visibly distinct and
lets the candidate-branch analysis read the relevant index directly from an
output tuple.  It is convenient but not necessary for these two promise
classes.  Proposition~\ref{prop:compact}, after the detailed proof,
replaces it by a single bit and improves the exponent measured in terms of
the permutation-domain size.  We retain $x$ here to keep the construction
transparent.
\end{remark}

Figure~\ref{fig:permutation-split} summarizes which values are prescribed by
the promise and which are chosen by a completion.

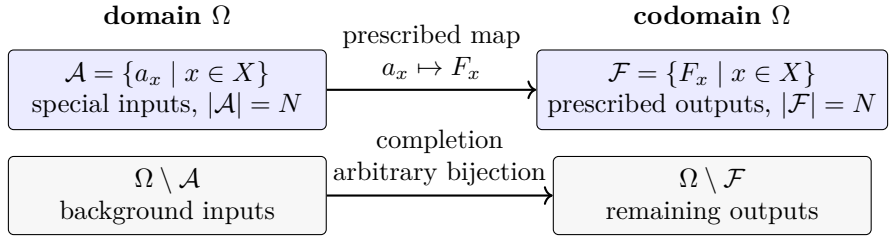
\begin{figure}[H]
\centering
\begin{tikzpicture}[
  block/.style={draw,rounded corners=2pt,minimum width=4.2cm,
                minimum height=1cm,align=center,inner sep=4pt},
  every node/.style={font=\small}
]
  \node at (0,1.65) {\textbf{domain $\Omega$}};
  \node at (7.2,1.65) {\textbf{codomain $\Omega$}};
  \node[block,fill=blue!8] (A) at (0,.65)
    {$\mathcal A=\{a_x\mid x\in X\}$\\special inputs, $|\mathcal A|=N$};
  \node[block,fill=black!3] (Ac) at (0,-.75)
    {$\Omega\setminus\mathcal A$\\background inputs};
  \node[block,fill=blue!8] (F) at (7.2,.65)
    {$\mathcal F=\{F_x\mid x\in X\}$\\prescribed outputs, $|\mathcal F|=N$};
  \node[block,fill=black!3] (Fc) at (7.2,-.75)
    {$\Omega\setminus\mathcal F$\\remaining outputs};
  \draw[->,thick] (A.east) --
    node[above,align=center] {prescribed map\\$a_x\mapsto F_x$} (F.west);
  \draw[->,thick] (Ac.east) --
    node[above,align=center] {completion\\arbitrary bijection} (Fc.west);
\end{tikzpicture}
\caption{The permutation promise fixes the map on the special inputs.  A
completion maps the background inputs bijectively to the remaining outputs.}
\label{fig:permutation-split}
\end{figure}

The prescribed map can be written as
\begin{equation}
   P(a_x)=F_x\qquad\text{for all }x\in X.
   \label{eq:permutation-prescription}
\end{equation}
For fixed $(h,r)$, a uniform completion is chosen from the $(L-N)!$
bijections between the background inputs and the remaining outputs.

Definition~\ref{def:standard-xor} specializes to
\[
   O_P\ket{w,z}=\ket{w,z\oplus P(w)}.
\]
The forward in-place interface is $U_P$ from
Definition~\ref{def:forward-in-place}.

The upper bounds are the same as for the injective problem.  For an XOR
query, we order the registers as the input register, followed by the marker,
$h$-value, input-copy, and tag target registers.
We prepare the input register in the uniform superposition of the states $\ket{a_x}$,
initialize the $h$-value target to $\ket{0^n}$, and initialize the marker,
input-copy, and tag targets to $\ket +$, $\ket +^{\otimes n}$, and
$\ket +^{\otimes n}$, respectively.  After one XOR query, we omit the target
fields that remain invariant and identify the input label $\ket{a_x}$ with
its address $\ket x$.  The relevant state is
\[
   \frac1{\sqrt N}\sum_{x\in X}\ket x\ket{h(x)}.
\]
As in Proposition~\ref{prop:xor-upper}, repeating the usual Simon sampling procedure $n+2$ times distinguishes the
two cases with error less than $1/4$: in the Simon case every sample lies in
$s^\perp$, whereas in the permutation case the samples are independent and
uniform in $X$, and a union bound gives
\(
  \Pr[\text{failure to span }X]<\frac14.
\)
A matching $O(\sqrt N)$ forward in-place upper bound follows by querying a
uniformly random set of $k=\Theta(\sqrt N)$ distinct special inputs and
looking for a repeated $h$-value (see Proposition~\ref{prop:erase-upper}). 

\subsection{Reverse coupling}
\label{subsec:reverse-coupling}

For each $i\in\{0,1\}$, the real hard distribution samples
$h\sim\Hcal_i$, a uniformly random tag table $r$, and then a uniform
completion $P$ satisfying~\eqref{eq:permutation-prescription}.  These three
objects remain fixed throughout the computation.  We call the resulting
execution the \emph{real experiment}.  The distribution of $P$ on background
inputs depends on the hidden set $\mathcal F$.  The following construction
couples $P$ to a uniform permutation independent of the hidden instance.

Let $\Omega^{(N)}$ be the set of ordered $N$-tuples
$W=(w_x)_{x\in X}$ of distinct elements of $\Omega$, and write
$\mathcal W=\{w_x\mid x\in X\}$.
For a permutation $\tau$ of $\Omega$, its support is
$\supp(\tau)=\{w\in\Omega:\tau(w)\ne w\}$.
For every $W\in\Omega^{(N)}$, fix a deterministic permutation $\tau_W$ such
that
\begin{equation}
   \tau_W(a_x)=w_x\quad\text{for all }x\in X,
   \qquad
   \supp(\tau_W)\subseteq\mathcal A\cup\mathcal W.
   \label{eq:tau-W-properties}
\end{equation}
The sets $\mathcal A$ and $\mathcal W$ may overlap, so the assignments
$a_x\mapsto w_x$ cannot in general be implemented by independent pairwise
swaps.  A suitable permutation nevertheless always exists.  The directed
edges $a_x\to w_x$
decompose into vertex-disjoint directed cycles and directed paths.  We retain
the cycle edges, close each nontrivial path by mapping its final vertex to
its initial vertex, and fix every point outside
$\mathcal A\cup\mathcal W$.  We fix one such permutation $\tau_W$ for each
$W$.

The purpose of $\tau_W$ is to relocate the inputs at which the outputs are
prescribed.  If $\pi=P\circ\tau_W^{-1}$, then
\[
   \pi(w_x)=P(a_x)=F_x.
\]
Thus the prescribed outputs are attached to the uniformly random ordered
tuple $W$ rather than to the fixed set $\mathcal A$.  The next lemma shows
that this random relocation makes $\pi$ uniform without changing the
conditional distribution required for $P$.

\begin{lemma}[Reverse coupling]
\label{lem:reverse-coupling}
Fix a family $F=(F_x)_{x\in X}$ arising from some pair $(h,r)$.  We sample $P$ uniformly from the
permutations satisfying~\eqref{eq:permutation-prescription}, and
independently sample $W$ uniformly from $\Omega^{(N)}$.  We define
\[
   \pi=P\circ\tau_W^{-1}.
\]
For every such fixed $F$ and every permutation $\pi_0$ of $\Omega$,
$\Pr[\pi=\pi_0\mid F]=1/L!$.  Consequently, $\pi$ is uniform and, when $F$
is generated from $(h,r)$, independent of $(h,r)$.  We call $\pi$ the
\emph{reference permutation}.
\end{lemma}

\begin{proof}
Fix a permutation $\pi_0:\Omega\to\Omega$.  If
$P\circ\tau_W^{-1}=\pi_0$, then $P=\pi_0\circ\tau_W$.  Evaluating this
identity at $a_x$ gives
\[
   F_x=P(a_x)=\pi_0(w_x),
\]
so necessarily
\[
   w_x=\pi_0^{-1}(F_x)\quad\text{for all }x\in X,
   \qquad
   P=\pi_0\circ\tau_W.
\]
Conversely, these choices are valid: the points $w_x$ are distinct because
the $F_x$ are distinct, and
\[
   P(a_x)=\pi_0(\tau_W(a_x))=\pi_0(w_x)=F_x.
\]
They also satisfy $P\circ\tau_W^{-1}=\pi_0$.  Hence exactly one pair $(P,W)$
in the support of the sampling experiment produces each $\pi_0$.  There are
$(L-N)!$ completions and
\[
   |\Omega^{(N)}|=(L)_N=\frac{L!}{(L-N)!}
\]
ordered distinct tuples.  Therefore
\[
   \Pr[\pi=\pi_0\mid F]
      =\frac1{(L-N)!}\frac1{(L)_N}
      =\frac1{L!},
\]
which is independent of $F$.
\end{proof}

Equivalently, $P=\pi\circ\tau_W$, where $\pi$ is uniform and independent
of $(h,r)$, while $W=(\pi^{-1}(F_x))_{x\in X}$ and hence $\tau_W$ may
depend on the hidden instance.  The relocation $\tau_W$ is supported on
at most $2N$ elements, and $W$ and $\pi$ remain fixed throughout the
coupled execution.  The hybrid for the signed functions defined next
converts this sparse modification into a bound on the final states.

\subsection{Signed simulation}
\label{subsec:signed-simulation}

We introduce functions whose inputs include a sign selecting the forward
or inverse direction.  They are fixed throughout each execution.  Their
\emph{disagreement set} consists of the signed inputs on which their
outputs differ.  Conditional on the real oracle, each fixed signed input
belongs to this set with probability at most $N/L$ over the coupling
randomness, as proved below.

For a permutation $P$, define the signed function
\[
   f_P(+,w)=P(w),
   \qquad
   f_P(-,w)=P^{-1}(w).
\]
We identify the two signs with one bit.  Since
$\Omega\cong\{0,1\}^{3n+1}$, the signed input $(\xi,w)$ is thereby encoded
as an element of $\{0,1\}^{3n+2}$, as required by the standard XOR model.
The sign $\xi\in\{+,-\}$ is an input of the XOR oracle used in this
simulation.  The ordered registers are the sign $\Xi$, the $\Omega$-valued
query register $Y$, and an $\Omega$-valued cleanup register $G$ that serves as the XOR
target.  Thus
\[
 O_{f_P}\ket{\xi,w,z}
   =\ket{\xi,w,z\oplus f_P(\xi,w)}.
\]
Sparse-disagreement hybrids are most naturally stated for standard XOR
oracles.  We therefore represent one forward in-place call by two XOR calls:
the first computes $P(w)$, a swap exchanges registers $Y$ and $G$ so that
$P(w)$ becomes the new query value, and the
second computes $P^{-1}(P(w))=w$ and removes the old input from $G$.  The sign
is fixed to $+$ for the first call and to $-$ for the second, so $\Xi$ is
suppressed.  The exact representation is
\begin{equation}
   \ket{w,0}_{Y,G}
      \xrightarrow{\ f_P(+,\cdot)\ }
   \ket{w,P(w)}
      \xrightarrow{\ \mathrm{SWAP}\ }
   \ket{P(w),w}
      \xrightarrow{\ f_P(-,\cdot)\ }
   \ket{P(w),0}_{Y,G}.
   \label{eq:signed-simulation}
\end{equation}
The lower-bound analysis applies this identity to each supplied forward
call.  The algorithm's oracle access remains $U_P$.  Hence a $T$-query
forward strategy has an exact representation using $q=2T$ signed XOR
queries.

Under the reverse coupling of Lemma~\ref{lem:reverse-coupling}, define the
comparison signed function
\begin{equation}
\begin{aligned}
   g(+,w)&=
   \begin{cases}
      F_x,&w=a_x,\\
      \pi(w),&w\notin\mathcal A,
   \end{cases}
   \\
   g(-,y)&=
   \begin{cases}
      a_x,&y=F_x,\\
      \pi^{-1}(y),&y\notin\mathcal F.
   \end{cases}
\end{aligned}
\label{eq:ideal-signed-function}
\end{equation}
Although the two directions in $g$ need not be inverses, each is a classical
function and therefore defines a unitary XOR transformation.
We call the computation that uses $g$ in both signed calls of every query
slot the \emph{ideal experiment}.

Conditioning on $P$ fixes the real oracle and every real pre-query state,
while $W$ remains independent and uniform.  The next lemma bounds the
marginal probability that a fixed signed input lies in the resulting
disagreement set.

\begin{lemma}[Sparse disagreement]
\label{lem:sparse-disagreement}
Fix $F$ and use the reverse coupling of
Lemma~\ref{lem:reverse-coupling}.  We condition on the completion $P$, and let
$\Gamma_W\subseteq\{+,-\}\times\Omega$ be the set on which $f_P$ and $g$ differ.
Because $W$ was sampled independently of $P$, it remains uniform in
$\Omega^{(N)}$ after this conditioning.  Then
\begin{equation}
   \max_{\xi\in\{+,-\},\ w\in\Omega}\,
      \Pr_W[(\xi,w)\in\Gamma_W\mid F,P]
      \leq \frac NL.
   \label{eq:sparse-disagreement-bound}
\end{equation}
\end{lemma}

\begin{proof}
On a forward input $a_x$, both functions output $F_x$.  If
$w\notin\mathcal A$, then
\[
   f_P(+,w)=P(w),
   \qquad
   g(+,w)=\pi(w)=P(\tau_W^{-1}(w)).
\]
Since $P$ is injective, disagreement implies that $\tau_W^{-1}(w)\neq w$.
A permutation and its inverse have the same support, so
\eqref{eq:tau-W-properties} and $w\notin\mathcal A$ imply
$w\in\mathcal W\setminus\mathcal A$.

On an inverse input $F_x$, both functions output $a_x$.  If
$y\notin\mathcal F$, then
\[
   f_P(-,y)=P^{-1}(y),
   \qquad
   g(-,y)=\pi^{-1}(y)
      =\tau_W(P^{-1}(y)).
\]
Disagreement implies $P^{-1}(y)\in\supp(\tau_W)$.  It cannot lie in
$\mathcal A$, because $y\notin\mathcal F=P(\mathcal A)$.  Hence
$P^{-1}(y)\in\mathcal W\setminus\mathcal A$.

After $P$ is fixed, both $w$ and $P^{-1}(y)$ in the preceding two cases are
fixed points of $\Omega$.  Every fixed point belongs to the underlying set
of a uniform ordered distinct $N$-tuple with probability $N/L$.  The
disagreement probability is zero on $(+,a_x)$ and $(-,F_x)$ and at most
$N/L$ on every other signed input, proving
\eqref{eq:sparse-disagreement-bound}.
\end{proof}

The next lemma permits arbitrary replacement values on a random disagreement
set.  Only the maximum probability that a fixed input belongs to it matters.

\begin{lemma}[Sparse-disagreement hybrid]
\label{lem:sparse-disagreement-hybrid}
Let $f:\{0,1\}^k\to\{0,1\}^{\ell}$ be fixed, and let a random seed $\omega$
determine a function $f_\omega:\{0,1\}^k\to\{0,1\}^{\ell}$ and a set
$\Gamma_\omega\subseteq\{0,1\}^k$ such that
$f_\omega(z)=f(z)$ for $z\notin\Gamma_\omega$.  Suppose
\[
   \max_{z\in\{0,1\}^k}\Pr_\omega[z\in\Gamma_\omega]\leq\eta.
\]
For a fixed algorithm making $q$ standard XOR queries, let $\rho_f$ be its
final state with oracle $f$, let
$\rho_{f_\omega}$ be its final state with oracle $f_\omega$, and set
$\bar\rho=\mathbb E_\omega\rho_{f_\omega}$.  Then
\begin{equation}
   \|\rho_f-\bar\rho\|_1\leq4q\sqrt\eta.
   \label{eq:sparse-disagreement-hybrid-bound}
\end{equation}
\end{lemma}

\begin{proof}
We purify the algorithm.  For $j=0,\ldots,q$, we define a hybrid in which the first
$j$ oracle calls use $f$ and the remaining $q-j$ calls use $f_\omega$.  We compare
the hybrids with $j-1$ and $j$ base-oracle calls.  Immediately before their
$j$th query, both have used the fixed oracle $f$ for their first $j-1$ calls.
Their common pre-query state $\ket{\psi_j}$ is therefore independent of $\omega$.

Let $\Pi_{\Gamma_\omega}$ project the query-input register onto
$\Gamma_\omega$.  Since the two XOR oracles agree outside
$\Gamma_\omega$ and are both unitary, their action changes
$\ket{\psi_j}$ by Euclidean norm at most
\[
   2\|\Pi_{\Gamma_\omega}\psi_j\|.
\]
We write $\ket{\psi_j}=\sum_z\ket z\ket{\phi_z}$.  Independence gives
\[
   \mathbb E_\omega\|\Pi_{\Gamma_\omega}\psi_j\|^2
      =\sum_z\Pr_\omega[z\in\Gamma_\omega]\|\phi_z\|^2
      \leq\eta.
\]
By Jensen's inequality, the expected contribution of one hybrid step is at
most $2\sqrt\eta$.  A telescoping sum therefore bounds the expected
Euclidean distance between the two final pure states by
$2q\sqrt\eta$.

For normalized vectors $\ket\phi,\ket\psi$,
\[
   \|\ket\phi\!\bra\phi-\ket\psi\!\bra\psi\|_1
      \leq2\|\phi-\psi\|.
\]
Taking expectations and using convexity of the trace norm proves
\eqref{eq:sparse-disagreement-hybrid-bound}.

\end{proof}

\begin{corollary}[Completion-to-ideal reduction]
\label{cor:completion-to-ideal}
Fix $h$ and $r$, and consider a $T$-query forward in-place strategy.  In the
signed representation, let $G_j$ be the fresh cleanup register used for
query slot $j$, and retain all of these registers in
\[
   \mathsf S_T=\mathsf{Alg}\otimes G_1\otimes\cdots\otimes G_T.
\]
After query slot $j$, extend every later inter-query operation by the identity
on $G_j$.
Let $\rho^{\mathrm{real}}_{h,r}$ be the final state of the original strategy,
averaged over a uniform completion $P$ and augmented by
$\ket0\!\bra0_{G_1}\otimes\cdots\otimes\ket0\!\bra0_{G_T}$.  Let
$\rho^{\mathrm{ideal}}_{h,r}$ be the final state obtained by replacing the
signed function $f_P$ by $g$ in every simulated query, retaining every
$G_j$, and averaging over the coupled pair $(P,W)$ from
Lemma~\ref{lem:reverse-coupling}.  Both states act on $\mathsf S_T$, and
\[
   \|\rho^{\mathrm{real}}_{h,r}
      -\rho^{\mathrm{ideal}}_{h,r}\|_1
      \leq 8T\sqrt{\frac NL}.
\]
\end{corollary}

\begin{proof}
Fixing $h$ and $r$ fixes $F$.  We condition on $P$.  By
the sampling rule in Lemma~\ref{lem:reverse-coupling}, $W$ remains uniform
in $\Omega^{(N)}$.  Lemma~\ref{lem:sparse-disagreement} therefore
bounds the probability of disagreement at every fixed signed input by
$\eta=N/L$.  The exact
representation~\eqref{eq:signed-simulation} uses $q=2T$ signed XOR calls.
Lemma~\ref{lem:sparse-disagreement-hybrid} gives
\[
   4q\sqrt\eta=8T\sqrt{\frac NL}.
\]
The estimate is uniform in $P$.  Averaging over the uniform-completion
marginal and using convexity therefore preserves the bound.
\end{proof}

\subsection{Query macro}
\label{subsec:query-macro}

Let $Y$ denote the $\Omega$-valued algorithm query register, both before and
after a query.  Each query slot uses a fresh $\Omega$-valued
cleanup register $G$, initially in $\ket0$.  It temporarily stores the old
query value between the two signed calls.  For $f_P$, the second call returns
$G$ to $\ket0$.  For the comparison function $g$, it may leave a nonzero
value, which is retained as part of the output.

We call the three-step composition consisting of the positive signed call,
the swap, and the negative signed call a \emph{query macro}.  When both calls
use $g$, the composition is the \emph{ideal query macro} $M$.  It maps $Y$ to
the ordered output $(Y,G)$.  The original algorithm continues on $Y$ and its
private workspace.  Later
operations act as the identity on this slot's $G$.  Retaining $G$ makes the
comparison map isometric, and discarding all cleanup registers at the end
can only decrease trace distance.

Combining the two comparison signed calls and the intervening swap in
\eqref{eq:signed-simulation} gives the following forward linear map:
\[
   M\ket w_Y
      =\ket{g(+,w)}_Y
         \ket{w\oplus g(-,g(+,w))}_G.
\]
Substituting~\eqref{eq:ideal-signed-function} gives
\begin{equation}
M\ket w=
\begin{cases}
   \ket{F_x}\ket0_G,
      &w=a_x,\\[1mm]
   \ket{F_x}\ket{w\oplus a_x}_G,
      &w\notin\mathcal A\text{ and }\pi(w)=F_x,\\[1mm]

      \ket{\pi(w)}\ket0_G,
      &w\notin\mathcal A\text{ and }\pi(w)\notin\mathcal F.
\end{cases}
\label{eq:ideal-macro}
\end{equation}
In the first line, the main output $Y$ reproduces the embedded injective
instance on a special input.  In the third, a background input is mapped by
the reference permutation $\pi$, whose distribution is uniform and
independent of $(h,r)$.  The second line is the only
place where a background query depends on the prescribed family
$F=(F_x)_{x\in X}$, and hence on $(h,r)$.  This dependence is forced
by the attempted inverse cleanup.  If a background input $w$ and a
special input $a_x$ have the same main-output value $F_x$, then
$w\oplus a_x\neq0$ in $G$ makes their complete $(Y,G)$ outputs orthogonal.
Thus $G$ retains the residual information that makes the map an isometry.

\begin{lemma}[Isometry of the ideal macro]
The map $M$ in~\eqref{eq:ideal-macro} is an isometry.
\end{lemma}

\begin{proof}
We show that distinct standard-basis inputs have orthogonal outputs.  If
$w,w'\notin\mathcal A$ and $w\neq w'$, their main-output values in $Y$ are
the distinct strings $\pi(w)$ and $\pi(w')$.  If $a_x\neq a_{x'}$ are
special inputs, their
main-output values $F_x$ and $F_{x'}$ are distinct because those strings contain $x$ and
$x'$, respectively.

It remains to compare a special input $a_x$ with a background input
$w\notin\mathcal A$.  If $\pi(w)\neq F_x$, their $Y$-values are different.  If
$\pi(w)=F_x$, their outputs are
\[
   \ket{F_x}\ket0_G
   \quad\text{and}\quad
   \ket{F_x}\ket{w\oplus a_x}_G.
\]
Since $w\neq a_x$, the two $G$-labels are distinct.  Thus all images of
standard-basis inputs are orthonormal, and $M^\dagger M=\Id$.
\end{proof}

By Lemma~\ref{lem:reverse-coupling}, $\pi$ is uniform and independent of
$(h,r)$.  Nevertheless, the case distinction
$\pi(w)\in\mathcal F$ between the second and third lines
of~\eqref{eq:ideal-macro} depends on the pair $(h,r)$ through $F$.  Thus Corollary~\ref{cor:injective-from-recording} is not
directly applicable: it concerns the injective query interface, which has no
background inputs.  We next construct a nearby contraction for the full
ideal macro and verify the hypotheses of
Theorem~\ref{thm:local-recording}.

\subsection{Recording contraction}
\label{subsec:recording-contraction}

Fix $h$ and $\pi$.  We purify the uniformly random fixed tag table using the registers $R_i$ and
the state $\ket U$ from Section~\ref{sec:recording-theorem}.  For a basis
table $\ket r$, the ideal macro uses the same $r$ on every call and preserves
that basis label.  With input order $(Y,\mathcal R)$ and output order
$(Y,G,\mathcal R)$, the coherently extended map acts as
\[
 M\ket w\ket r=
 \begin{cases}
  \ket{1,h(i),i,r_i}_Y\ket0_G\ket r,
     &w=a_i,\\[1mm]
  \ket{\pi(w)}_Y\ket{w\oplus a_i}_G\ket r,
     &w\notin\mathcal A\text{ and }
       \pi(w)=(1,h(i),i,r_i)\text{ for some }i,\\[1mm]
  \ket{\pi(w)}_Y\ket0_G\ket r,
     &\text{otherwise.}
\end{cases}
\]
For each $i$, the projectors $\Pi_{0,i}=\ket u\!\bra u$ and
$\Pi_{1,i}=\Id-\Pi_{0,i}$ act on $R_i$ and select its unrecorded and recorded
subspaces, respectively.
The algorithm workspace is an unchanged tensor factor.  Thus
$M$ does not overwrite a computational-basis table label, but it is
coherently controlled by that label and may entangle $\mathcal R$ with
$(Y,G)$.  Basis-wise preservation of $\ket r$ therefore does not imply that
$M$ preserves the decomposition
$\operatorname{im}\Pi_{0,i}\oplus\operatorname{im}\Pi_{1,i}$ on a superposition
of table values.  This distinction will explain the post-projection in
\eqref{eq:Mtilde-candidate}.

Let
\[
   \Omega_1
     =\{(1,v,i,t)\mid v,i,t\in X\}.
\]
$\Omega_1$ is the subset of $\Omega$ whose marker bit is one.  Every
$y\in\Omega_1$ has a unique parsing $y=(1,v,i,t)$, and
\[
   y=F_i
   \quad\Longleftrightarrow\quad
   v=h(i)\text{ and }r_i=t.
\]

Let $w\notin\mathcal A$ and set $y=\pi(w)$.  We call the branch labelled by
$w$ a \emph{candidate branch for index $i$} if
\[
   y=(1,h(i),i,t)
   \quad\text{for some }t\in X.
\]
The index $i$ is parsed from the main
output $y$, not from $w$: a general background input carries no designated
address.  For every background branch, define the table-independent baseline
isometry
\begin{equation}
   M_w^{\mathrm{base}}\ket\psi_{\mathcal R}
      =\ket y_Y\ket0_G\ket\psi_{\mathcal R}.
   \label{eq:candidate-baseline}
\end{equation}
This is the action of the macro when no equality $y=F_i$ is detected.

The ideal macro tests whether $y\in\mathcal F$.  Once $y=(1,v,i,t)$ is
parsed, membership is the conjunction of the equalities $v=h(i)$ and
$r_i=t$.  If $v\neq h(i)$, no value of
$r_i$ can make $y=F_i$, and the macro acts on branch $w$ by the baseline
for every table.  If
$v=h(i)$, membership reduces to the coherent
equality test $r_i=t$.  This is the candidate case modified below.  We will
define $\widetilde M$ so that its action on
$\operatorname{im}\Pi_{0,i}$ is equal to the baseline in both cases, so the
unrecorded component cannot acquire dependence on $h(i)$ or $r_i$.

We define a linear map $\widetilde M$ separately on each
standard-basis query-input branch.
Formally, let $\iota_w\ket\psi=\ket w_Y\ket\psi_{\mathcal R}$ and set
$M_w=M\iota_w$.  We now specify maps $\widetilde M_w$ from the table system to
the ordered output $(Y,G,\mathcal R)$.  Since the states $\ket w_Y$ form a
basis, the relations $\widetilde M\iota_w=\widetilde M_w$ uniquely determine
$\widetilde M$.  All other registers are spectators.

\begin{enumerate}
\item If the input is the special input $a_i$, set
\begin{equation}
    \widetilde M_{a_i}=\Pi_{1,i}M_{a_i}.
    \label{eq:Mtilde-special}
\end{equation}

\item Let the input be $w\notin\mathcal A$ and set $y=\pi(w)$.  If
$y\notin\Omega_1$, or if $y=(1,v,i,t)$ with $v\neq h(i)$, set
\[
    \widetilde M_w=M_w.
\]
In this case $y$ cannot be a prescribed output, so
$M_w=M_w^{\mathrm{base}}$.

\item Suppose that $w\notin\mathcal A$ and
\[
    y=\pi(w)=(1,h(i),i,t).
\]
Using the baseline $M_w^{\mathrm{base}}$ from~\eqref{eq:candidate-baseline}, define
\begin{equation}
   \widetilde M_w
      =M_w^{\mathrm{base}}\Pi_{0,i}+\Pi_{1,i}M_w\Pi_{1,i}.
   \label{eq:Mtilde-candidate}
\end{equation}
On the unrecorded component $\Pi_{0,i}$, this rule retains the
table-independent baseline and suppresses the equality test $r_i=t$.  On the
recorded component $\Pi_{1,i}$, it applies $M_w$ and projects the output
back into the recorded subspace.  Dependence on $h(i)$ is allowed on this
component because cell $i$ is already recorded.
The right-hand $\Pi_{1,i}$ restricts $M_w$ to an input component in which
cell $i$ is already recorded.  The left-hand $\Pi_{1,i}$ keeps the resulting
output in the recorded subspace.  This output projection is needed because
$M_w$ uses $r_i$ as a
coherent control and can change the interference between basis values.  For
example, choose $t'\neq t$, let $d=w\oplus a_i\neq0$, and suppress the
common $Y$ factor and all table cells other than $R_i$.  Then
\[
   \ket\psi=\frac{\ket t-\ket{t'}}{\sqrt2}
      \in\operatorname{im}\Pi_{1,i},
   \qquad
   \Pi_{0,i}M_w\ket\psi
      =\frac{\ket d_G-\ket0_G}{\sqrt{2N}}\ket u_{R_i}\neq0.
\]
The equality-controlled action has moved part of a recorded input into the
unrecorded subspace.  The left projector removes this leakage.  Since
$M_w^{\mathrm{base}}$
acts trivially on the table,
$M_w^{\mathrm{base}}\Pi_{0,i}
 =\Pi_{0,i}M_w^{\mathrm{base}}\Pi_{0,i}$, so
\eqref{eq:Mtilde-candidate} is block diagonal in
$\Pi_{0,i}\oplus\Pi_{1,i}$.
\end{enumerate}
Projectors in these formulas act on the selected table-cell register $R_i$
and are tensored with the identity on all other registers.

The action on a basis table makes all three cases explicit.
For the special and candidate displays we regroup
$\mathcal R=R_i\otimes\mathcal R_{\ne i}$ and use the output order
$(Y,G,R_i,\mathcal R_{\ne i})$.
We write
\[
   y_{i,t}=(1,h(i),i,t)
      \quad\text{for }t\in X,
   \qquad
   \ket{p_t}=\Pi_{1,i}\ket t,
   \quad \|p_t\|^2=1-\frac1N,
   \qquad
   \ket{r_{\ne i}}=\bigotimes_{j\ne i}\ket{r_j}.
\]
On a special input,
\[
 \widetilde M\ket{a_i}\ket r
   =\ket{y_{i,r_i}}_Y\ket0_G
      \bigl(\Pi_{1,i}\ket{r_i}\bigr)_{R_i}\ket{r_{\ne i}}.
\]
On a noncandidate background input with $y=\pi(w)$,
\[
 \widetilde M\ket w\ket r
   =\ket y_Y\ket0_G\ket r_{\mathcal R}.
\]
Finally, on a candidate branch
$y=(1,h(i),i,t)$, let $d=w\oplus a_i$.  Then
\begin{align*}
 \widetilde M\ket w\ket r
 =\ket y_Y\Bigl[&
   \ket0_G\ket{r_i}_{R_i}
   +\left(\mathbf 1_{\{r_i=t\}}-\frac1N\right)
      \bigl(\ket d_G-\ket0_G\bigr)\ket{p_t}_{R_i}
   \Bigr]\ket{r_{\ne i}}.
\end{align*}
The projectors may therefore produce a superposition inside $R_i$ even when
the input table is a basis state.

The baseline term on the unrecorded subspace is essential.  For
$\ket{u_{\ne t}}=(N-1)^{-1/2}\sum_{\zeta\in X\setminus\{t\}}\ket\zeta$,
with all other table cells suppressed, the unmodified candidate-branch map
$M_w$ agrees with $M_w^{\mathrm{base}}$, so
\[
   \|\Pi_{0,i}M_w\ket{u_{\ne t}}\|
      =\sqrt{\frac{N-1}{N}}.
\]
Thus the output component whose $R_i$ factor lies in
$\operatorname{im}\Pi_{0,i}$ can have norm close to one.  Only its
table-dependent defect from the baseline is small:
\[
   \|(M_w-M_w^{\mathrm{base}})\Pi_{0,i}\|=\sqrt{\frac2N}.
\]
Equation~\eqref{eq:Mtilde-candidate} retains this baseline component and removes
only its table-dependent defect.  In this respect the construction differs
from the injective recording projection.

In the next lemma, the contraction bound does not follow from contraction
of each fixed branch $\widetilde M_w$: different input branches can have
overlapping output ranges and interfere coherently.

\begin{lemma}[One-query recording estimate]
\label{lem:permutation-one-query-recording}
For every fixed $h$ and $\pi$,
\[
   \|M-\widetilde M\|\leq\frac4{\sqrt N},
   \qquad
   \|\widetilde M\|\leq1.
\]
\end{lemma}

\begin{proof}
\noindent\emph{Approximation.}
On a special input $a_i$, the
macro copies the basis value $r_i$ of the table-cell register into the tag
field of the main output $Y$.  The same one-cell calculation used in
Section~\ref{sec:injective-lower-bound} gives
\[
   \|(M_{a_i}-\widetilde M_{a_i})\|
      =\|\Pi_{0,i}M_{a_i}\|
      =\frac1{\sqrt N}.
\]

Now consider a candidate branch labelled by $w$ with
$\pi(w)=(1,h(i),i,t)$.  Let $\Lambda_t=\ket t\!\bra t$ on $R_i$.  The difference
$M_w-M_w^{\mathrm{base}}$ is supported on $\Lambda_t$ and changes the
$G$-register from $\ket0$ to
the distinct basis state $\ket{w\oplus a_i}$.  Explicitly, on a basis table,
\[
 (M_w-M_w^{\mathrm{base}})\ket r=
 \begin{cases}
   0,&r_i\ne t,\\[1mm]
   \ket y_Y\bigl(\ket{w\oplus a_i}_G-\ket0_G\bigr)\ket r,
      &r_i=t.
 \end{cases}
\]
Hence
$M_w-M_w^{\mathrm{base}}=(M_w-M_w^{\mathrm{base}})\Lambda_t$ and
\[
   \|(M_w-M_w^{\mathrm{base}})\Pi_{0,i}\|
      =\sqrt2\,\|\Lambda_t\Pi_{0,i}\|
      =\sqrt{\frac2N}.
\]
Since $M_w^{\mathrm{base}}$ acts trivially on the table,
$\Pi_{0,i}M_w^{\mathrm{base}}\Pi_{1,i}=0$.  Moreover,
$M_w-M_w^{\mathrm{base}}$ is the tensor product of the
output vector
$\ket y_Y(\ket{w\oplus a_i}_G-\ket0_G)$, of norm $\sqrt2$, with $\Lambda_t$ on
$R_i$.  The rank-one operator $\Pi_{0,i}\Lambda_t\Pi_{1,i}$ has norm
$N^{-1/2}\sqrt{1-1/N}$.  Therefore
\[
   \|\Pi_{0,i}M_w\Pi_{1,i}\|
      =\|\Pi_{0,i}(M_w-M_w^{\mathrm{base}})\Pi_{1,i}\|
      =\sqrt{\frac2N}\sqrt{1-\frac1N}
      \leq\sqrt{\frac2N}.
\]
Using
\[
   M_w-\widetilde M_w
      =(M_w-M_w^{\mathrm{base}})\Pi_{0,i}+\Pi_{0,i}M_w\Pi_{1,i},
\]
we obtain
\[
   \|M_w-\widetilde M_w\|
      \leq2\sqrt{\frac2N}.
\]
Distinct background inputs have distinct main-output values $\pi(w)$ in $Y$, so their
error maps form a direct sum and the same bound holds for all background
branches together.  Likewise, the special error branches for distinct $i$
have orthogonal address fields in $Y$.  The special and background error
ranges need not be mutually orthogonal, so the triangle inequality between
those two aggregate maps gives
\[
   \|M-\widetilde M\|
      \leq\frac{1+2\sqrt2}{\sqrt N}
      <\frac4{\sqrt N}.
\]

\medskip\noindent\emph{Contraction.}
The only possible coherent overlap is between the special branch $a_i$ and
candidate background branches $w_t$ whose main output is $y_{i,t}$.  The
projection on the special branch and the modification of the candidate
branches reduce their diagonal squared norms.  The four Gram identities
below show that these losses dominate the cross terms for an arbitrary
coherent superposition.  Fix $i\in X$.  For $t\in X$, let
$y_{i,t}=(1,h(i),i,t)$ and $w_t=\pi^{-1}(y_{i,t})$, and define the candidate set
\[
  \mathcal C_i
   =\{t\in X\mid w_t\notin\mathcal A\}.
\]
For $t\in\mathcal C_i$, write $\ket{p_t}=\Pi_{1,i}\ket t$.
We define the associated query-input subspace
\[
   \mathcal B_i
      =\operatorname{span}
        \bigl(\{\ket{a_i}\}\cup
        \{\ket{w_t}:t\in\mathcal C_i\}\bigr).
\]
We call $\mathcal B_i$, tensored with the table and spectator registers, the
$i$-block.  Let
$V_i=\widetilde M_{a_i}$ and $K_{i,t}=\widetilde M_{w_t}$.  The basis actions
above give, for $t,t'\in\mathcal C_i$,
\begin{align*}
 V_i^\dagger V_i
    &=\left(1-\frac1N\right)\Id,\\
 K_{i,t}^\dagger K_{i,t}
    &=\Id-\frac2N\ket{p_t}\!\bra{p_t},\\
 V_i^\dagger K_{i,t}
    &=\frac1N\ket t\!\bra{p_t},\\
 K_{i,t}^\dagger K_{i,t'}&=0
    \qquad\text{for }t\neq t'.
\end{align*}
The first identity measures the uniform component removed from the special
branch.  The second measures the loss caused by keeping the candidate branch
block diagonal, the third is the remaining special--candidate overlap, and
the fourth follows from the distinct main outputs $y_{i,t}$.  All four
identities remain valid after tensoring with the algorithm workspace and the
other table cells.  Appendix~\ref{app:gram-identities} derives them directly
from the basis actions.

An arbitrary input supported on this block has the form
\[
   \ket{a_i}_Y\ket\gamma
     +\sum_{t\in\mathcal C_i}\ket{w_t}_Y\ket{\zeta_t},
\]
where $\ket\gamma$ and the $\ket{\zeta_t}$ are vectors on the table and
spectator registers.  We define
\[
   c=\sum_{t\in\mathcal C_i}
       \bigl(\ket t\!\bra{p_t}\otimes\Id\bigr)\ket{\zeta_t}.
\]
The orthogonality of the $\ket t$ gives
$\|c\|^2=\sum_{t\in\mathcal C_i}
\|(\bra{p_t}\otimes\Id)\ket{\zeta_t}\|^2$.  The four Gram identities show that
the squared output norm minus the squared input norm on this entire block is
\[
   -\frac1N\|\gamma\|^2-\frac2N\|c\|^2
      +\frac2N\operatorname{Re}\langle\gamma,c\rangle
   =-\frac1N\bigl(\|\gamma-c\|^2+\|c\|^2\bigr)
   \leq0.
\]
Thus the norm removed by the recording projections dominates every
special--candidate interference term.  The complete norm expansion and
completed-square calculation appear in
Appendix~\ref{app:contraction-calculation}.

Blocks with different addresses have orthogonal $Y$-outputs.  Every
remaining background branch also has a $Y$-output orthogonal to these
blocks.  On those branches $\widetilde M=M$, and distinct inputs have distinct
$Y$-outputs because $\pi$ is a permutation.  The full map is therefore an
orthogonal sum of contractions, which proves $\|\widetilde M\|\leq1$.
\end{proof}

\subsection{Locality}
\label{subsec:permutation-locality}

We now verify the two structural hypotheses of
Theorem~\ref{thm:local-recording}.  For $D\subseteq X$, the table
subspace $\mathcal R_D$ consists of states in which exactly the cells indexed
by $D$ lie in their recorded subspaces, and $\Pi_D$ denotes its projector.
After $T$ query slots the retained system is
\[
   \mathsf S_T
      =\mathsf{Alg}\otimes G_1\otimes\cdots\otimes G_T.
\]
The recording global vector lies in $\mathsf S_T\otimes\mathcal R$.  Every
$\Pi_D$ below acts on $\mathcal R$, while the identity in
$(\Id\otimes\Pi_D)$ acts on $\mathsf S_T$.  Tracing out $\mathcal R$ retains
the algorithm system and all retained cleanup registers.

\begin{lemma}[Monotone recording and local dependence]
\label{lem:permutation-locality}
Fix $\pi$.  After $T$ queries in which every ideal macro $M$ is replaced by
$\widetilde M$, the final vector has a decomposition
\[
   \ket{\widetilde\Psi_{h,\pi}}
      =\sum_{\substack{D\subseteq X\\|D|\leq T}}
          \ket{\widetilde\Psi_{h,\pi,D}},
   \qquad
   \ket{\widetilde\Psi_{h,\pi,D}}
      =(\Id\otimes\Pi_D)\ket{\widetilde\Psi_{h,\pi}}.
\]
For every $D$, this vector depends on $h$ only through $h|_D$.
\end{lemma}

\begin{proof}
The table starts in the empty sector, and inter-query operations do not act on
it.  We consider one application of $\widetilde M$.

A special query to $a_i$ ends with $\Pi_{1,i}$ by
\eqref{eq:Mtilde-special}.  It therefore changes a sector $D$ to
$D\cup\{i\}$.  Its dependence on $h$ is only through $h(i)$.  If $i$ was already in
$D$, the sector remains $D$ and the same dependence is allowed because
$h(i)$ is part of $h|_D$.

We consider a background branch whose main output parses as
$\pi(w)=(1,v,i,t)$.  If $v\neq h(i)$, this is a noncandidate branch:
$\widetilde M_w$ is the table-independent baseline and preserves the sector.  If
$v=h(i)$, it is a candidate branch and is block diagonal in
$\Pi_{0,i}\oplus\Pi_{1,i}$ by~\eqref{eq:Mtilde-candidate}, so it also preserves
the sector.  When $i\notin D$, the cell lies in its $\Pi_{0,i}$ subspace, and
both alternatives use the same baseline action there.  Thus changing
$h(i)$ cannot change the action on a sector not containing $i$.  When
$i\in D$, the branch may depend on $h(i)$, but that value is contained in
$h|_D$.

Every remaining background branch is table independent and cannot introduce
dependence on an unrecorded value of $h$.  Consequently, one query adds at
most one index, and any dependence on $h(i)$ occurs only in a sector
containing $i$.  Induction over the query slots proves both assertions.
The induction is linear, so the same conclusion holds for the coherent sum
of all contributions to a sector.  Appendix~\ref{app:permutation-locality}
verifies the sector blocks explicitly.  Appendix~\ref{app:sector-induction}
gives the general induction.
\end{proof}

After tracing out the table, sectors with different labels have no cross
terms.  For fixed $\pi$, equation~\eqref{eq:sigma-D-varphi-definition}
therefore defines positive operators $\sigma(D,\varphi)$ on the retained system
$\mathsf S_T$, which includes all fresh $G$ registers.  Lemmas
\ref{lem:permutation-one-query-recording} and
\ref{lem:permutation-locality}, together with
Lemma~\ref{lem:simon-restrictions} for the restriction property, verify
the hypotheses of
Theorem~\ref{thm:local-recording} with
\[
   \mathcal Q_{h,\pi}=M,
   \qquad
   \widetilde{\mathcal Q}_{h,\pi}=\widetilde M,
   \qquad
   \delta=4N^{-1/2}.
\]
The parameter $\pi$ is uniform and independent of $h$ and $r$, exactly as
required by that theorem.

\begin{theorem}[Ideal trace-distance estimate]
\label{thm:ideal-trace-distance}
For $i\in\{0,1\}$, let $\rho^{\mathrm{ideal}}_i$ be the final state
averaged over $h\sim\Hcal_i$, one uniform fixed tag table, and an
independent uniform reference permutation $\pi$, when each query slot
uses the ideal macro~\eqref{eq:ideal-macro}.  Equivalently, initialize the
table purification in $\ket U$, run the coherently extended macros, trace
out $\mathcal R$, and average over $h$ and $\pi$.  The state acts on
$\mathsf S_T$, with every fresh $G$ register retained.  If
$4T/\sqrt N\leq1$, then
\[
   \|\rho^{\mathrm{ideal}}_0-
      \rho^{\mathrm{ideal}}_1\|_1
      \leq
      \frac{T(T-1)}{N-1}+\frac{24T}{\sqrt N}.
\]
\end{theorem}

\begin{proof}
We apply Theorem~\ref{thm:local-recording} with
$\delta=4N^{-1/2}$ and average
over the uniform permutation $\pi$, sampled independently of $h$ and $r$.
\end{proof}

\subsection{Final comparison}
\label{subsec:final-comparison}

For $i\in\{0,1\}$, let $\bar\rho_i$ be the classical average of the
strategy's original final state on $\mathsf{Alg}$ for $h\sim\Hcal_i$, one
uniformly random tag table sampled once, and one uniformly random completion
$P$ sampled once.  We define the augmented real state on
$\mathsf S_T=\mathsf{Alg}\otimes G_1\otimes\cdots\otimes G_T$ by
\[
  \rho_i^{\mathrm{real}}
    =\bar\rho_i\otimes
      \ket0\!\bra0_{G_1}\otimes\cdots\otimes\ket0\!\bra0_{G_T}.
\]
We represent this distribution by sampling the actual uniform completion
$P$, independently sampling a uniform ordered distinct tuple $W$, and
setting $\pi=P\circ\tau_W^{-1}$ as in
Lemma~\ref{lem:reverse-coupling}.  Thus the marginal of $P$ is exactly the
required real distribution, while the derived reference permutation $\pi$
is uniform and independent of $(h,r)$.  Equation~\eqref{eq:signed-simulation}
justifies this augmentation: in the real simulation every $G_j$
returns to the same unentangled state $\ket0$.  Consequently
$\Tr_{G_1\cdots G_T}\rho_i^{\mathrm{real}}=\bar\rho_i$ and
\[
  \|\rho_0^{\mathrm{real}}-\rho_1^{\mathrm{real}}\|_1
    =\|\bar\rho_0-\bar\rho_1\|_1.
\]

Corollary~\ref{cor:completion-to-ideal} holds for every fixed $h,r$.  After
averaging within each promise class, it gives
\begin{equation}
   \|\rho^{\mathrm{real}}_i-
      \rho^{\mathrm{ideal}}_i\|_1
      \leq8T\sqrt{\frac NL}
   \qquad\text{for }i\in\{0,1\}.
   \label{eq:real-ideal-each-class}
\end{equation}
The ideal state in this equation has exactly the distribution used in
Theorem~\ref{thm:ideal-trace-distance}, because the reverse-coupling lemma
makes the derived $\pi$ uniform and independent of $h$ and $r$.

\begin{theorem}[Permutation lower-bound estimate]
\label{thm:permutation-lower-bound-estimate}
Every $T$-query forward in-place strategy satisfies
\begin{equation}
   \|\rho^{\mathrm{real}}_0-
      \rho^{\mathrm{real}}_1\|_1
      \leq
      \frac{T(T-1)}{N-1}
      +\frac{24T}{\sqrt N}
      +16T\sqrt{\frac NL}
   \label{eq:permutation-final-quantitative}
\end{equation}
whenever $4T/\sqrt N\leq1$.
\end{theorem}

\begin{proof}
We use the triangle inequality, equation~\eqref{eq:real-ideal-each-class} for
the two promise classes, and
Theorem~\ref{thm:ideal-trace-distance}.
\end{proof}

For the present construction $L=2N^3$, so
\[
   16T\sqrt{\frac NL}
      =\frac{8\sqrt2\,T}{N}.
\]
Suppose $T\leq\sqrt N/100$.  The condition of
Theorem~\ref{thm:permutation-lower-bound-estimate} then holds, and for every
$N\geq2$ the three terms in~\eqref{eq:permutation-final-quantitative} are at
most
\[
   \frac1{5000},
   \qquad
   \frac{24}{100},
   \qquad
   \frac8{100},
\]
respectively.  Their sum is strictly less than $2/3$.

On the other hand, an algorithm that succeeds with probability at least
$2/3$ on every promised permutation also distinguishes the two equal-prior
hard distributions with success probability at least $2/3$.  Helstrom's
bound would then require
\[
   \|\rho^{\mathrm{real}}_0-
      \rho^{\mathrm{real}}_1\|_1\geq\frac23,
\]
a contradiction.  Therefore every bounded-error forward in-place algorithm
uses more than $\sqrt N/100$ queries.  Together with the collision-sampling
upper bound, this proves
\[
   Q_{\mathrm{IP}}(\mathrm{PermutationGarbageSimon}_n)
      =\Theta(\sqrt N)=\Theta(2^{n/2}).
\]
The XOR upper bound is at most $n+2$.  Since the permutation universe has
$L=2^{3n+1}$ elements, these bounds are $O(\log L)$ versus
$\Theta(L^{1/6})$.

This completes the proof of Theorem~\ref{thm:combined-separation}.

\subsection{Compact permutation problem}
\label{subsec:compact}

The explicit copy of $x$ in the prescribed outputs can be replaced by a
single bit.  A \emph{bit table} for $h$ is a function $b:X\to\{0,1\}$,
written $b=(b_x)_{x\in X}$, such that $\kappa(x)=(h(x),b_x)$ is injective.
Every function $b$ is a bit table for a permutation, and for a Simon
function with shift $s$ the condition is $b_{x\oplus s}=1-b_x$ for all $x$.
Let $\Omega'=\{0,1\}\times X\times\{0,1\}\times X$ and
$L'=|\Omega'|=4N^2$, and define
\[
   a'_x=(0,x,0,0^n),
   \qquad
   F'_x=(1,h(x),b_x,r_x).
\]
An instance of $\mathrm{CompactGarbageSimon}_n$ is a permutation $P'$
of $\Omega'$ satisfying $P'(a'_x)=F'_x$ for all $x\in X$, for a promised
$h$, a bit table $b$ for $h$, and a tag table $r$.  Every completion is
allowed.  The answer is NO when $h$ is a permutation and YES when $h$ is a
Simon function.  The values $F'_x$ are distinct because $\kappa$ is
injective, and their marker bits separate them from the special inputs.
Without knowledge of $h$ and $b$, the address $x$ is in general not
determined by $F'_x$.

\begin{proposition}[Compact permutation separation]
\label{prop:compact}
Let $n\geq1$ and $L'=4N^2$.  At most $n+2$ standard queries suffice to
solve $\mathrm{CompactGarbageSimon}_n$, and every bounded-error
forward in-place algorithm for it uses more than $\sqrt N/100$ queries.
Consequently, $Q_{\mathrm{IP}}(\mathrm{CompactGarbageSimon}_n)
=\Theta(2^{n/2})=\Theta((L')^{1/4})$.
\end{proposition}

\begin{proof}
\emph{Upper bounds.}  The standard-query and collision-sampling upper
bounds are unchanged because they use only the $h$-coordinate; in the XOR
algorithm, the bit target is also initialized to $\ket+$.

\emph{Hard distributions.}  The hidden data is the function
$x\mapsto(h(x),b_x)$, with $\Sigma=X\times\{0,1\}$.  We choose
$h\sim\Hcal_0$ or $h\sim\Hcal_1$ as before.  In the permutation case, the
bits are independent and uniform.  In the Simon case, independently for
each pair $\{x,x\oplus s\}$, one uniformly chosen endpoint receives bit $0$
and the other receives bit $1$.  The tag table and the completion are
uniform, as before.  Fix $D\subseteq X$.  In the permutation case, $h|_D$
is a uniformly random injection and $(b_x)_{x\in D}$ are independent
uniform bits, independent of $h$.  In the Simon case, conditional on
$s\notin\Delta(D)$, the elements of $D$ lie in distinct pairs, so the same
holds by the proof of Lemma~\ref{lem:simon-restrictions}.  Hence the
restriction property~\eqref{eq:restriction-property} holds with
$p_D=|\Delta(D)|/(N-1)$.

\emph{Lower bound.}  The permutation-specific parts of the proof require
only a corresponding change in how a candidate index is identified.  The
new candidate outputs are
\[
   y'_{i,t}=(1,h(i),b_i,t).
\]
Their prefixes $(h(i),b_i)=\kappa(i)$ are distinct, so a candidate output
determines a unique index $i$ for each fixed hidden instance.
The distinct prefixes also make the $Y$-output spaces of different
$i$-blocks mutually orthogonal, and likewise make the special-branch error
ranges for distinct indices orthogonal.
This index
need not be readable without knowledge of that instance, so locality
requires a separate check.  Fix an output $y'=(1,v,\beta,t)$ and a sector
$D$.  If $(v,\beta)$ lies in $\kappa(D)$, its unique index $i\in D$ and the
candidate action are determined by $\kappa|_D$.  Otherwise, the branch is
either noncandidate or a candidate for an index outside $D$.  Both cases
use the same baseline action on that sector.  Thus two hidden instances
agreeing on $(h,b)|_D$ induce the same block, even if their candidate
indices outside $D$ differ.  This proves the required locality.  The
reverse coupling uses only the distinctness of the prescribed outputs, and
the contraction and Gram calculations use only the distinctness of the
candidate labels and the unchanged $n$-bit tag states.  Consequently the
preceding proof applies with the hidden data $(h,b)$ and with $L$ replaced
by $L'=4N^2$.

In particular, the last term in
Theorem~\ref{thm:permutation-lower-bound-estimate} becomes
\[
   16T\sqrt{\frac{N}{L'}}=\frac{8T}{\sqrt N}.
\]
The same choice $T\leq\sqrt N/100$ therefore keeps the trace distance below
$2/3$ and gives an $\Omega(\sqrt N)$ lower bound.  Hence the compact
permutation problem has query complexities $O(\log L')$ and
$\Theta(\sqrt N)=\Theta((L')^{1/4})$ in the standard and forward in-place
models, respectively.
\end{proof}

\begin{remark}[Compact injective variant]
\label{rem:compact-injective}
The same compression gives an injective promise problem without an
explicit copy of the input in the output.  For a promised $h$, any bit
table $b$ for $h$, and any tag table $r$, set
\[
   f'_{h,b,r}(x)=(h(x),b_x,r_x)=(\kappa(x),r_x).
\]
Deciding whether $h$ is a permutation or a Simon function has the same
bounds: at most $n+2$ standard queries and $\Theta(\sqrt N)$
forward-erasing queries.

For the lower bound, let $V_\kappa$ be the purified forward-erasing query
and set
\[
   \Pi_{\mathrm{rec}}^\kappa
      =\sum_{x\in X}\ket{\kappa(x)}\!\bra{\kappa(x)}\otimes\Pi_{1,x},
   \qquad
   \widetilde V_\kappa=\Pi_{\mathrm{rec}}^\kappa V_\kappa,
\]
with identities on the output-tag and unaffected registers implicit.
Distinct prefixes $\kappa(x)$ make this an orthogonal projector and make
the error ranges for distinct addresses orthogonal.  The one-cell
calculation in Section~\ref{sec:injective-lower-bound} still gives
$\|V_\kappa-\widetilde V_\kappa\|=N^{-1/2}$.
Each query records its input $x$ and uses only $\kappa(x)$, verifying the
sector and locality conditions.  The hard distributions used above satisfy
the restriction property, so Theorem~\ref{thm:local-recording} gives the
lower bound.  Both upper bounds use only $h$ and remain unchanged.
\end{remark}

\section{Discussion and open problems}
\label{sec:combined-discussion}

Theorem~\ref{thm:local-recording} separates the reusable part of the
argument from the interface-specific construction of a recording
contraction.  It applies to other forward-only interfaces whenever the
fixed oracle randomness has a cellwise purification, each query has a
nearby contraction with the stated sector-transition and locality
properties, and the hard distributions satisfy the restriction
property~\eqref{eq:restriction-property}.  For permutations, the reverse
coupling and the candidate-branch recording estimate handle queries
outside the embedded instance, after which the sector and positivity
argument and the bound $\binom{T}{2}/(N-1)$ on the probability that a set
of at most $T$ addresses contains a Simon pair are reused unchanged.  We list some questions that our
results leave open.

\begin{itemize}
\item In Aaronson's original candidate, the garbage is long and random,
and injectivity holds only with high probability; no separate bit
enforces it.  Imposing injectivity as a promise makes the tags dependent
across addresses, whereas our purification uses independent cells.  Does
the lower bound extend to this candidate?

\item Aaronson suggested more generally that a sufficient amount of
garbage in an erasing oracle's responses should be equivalent, under
appropriate conditions, to decohering those responses.
Theorem~\ref{thm:local-recording} gives a statement of this kind for hard
distributions satisfying~\eqref{eq:restriction-property}.  A general
formulation remains open.

\item The conjectured separation for the candidate decision problem of
Holman, Ramachandran, and Yirka~\cite{Holman2025} remains open.

\item Our compact permutation problem has forward in-place complexity
$\Theta(L^{1/4})$ in terms of its domain size $L$.  Is there a decision
problem on permutations with polylogarithmic standard query complexity and
forward in-place complexity \(L^{c}\) for some \(c>1/4\)?
\end{itemize}

\clearpage
\appendix

\section{Standard facts and model conventions}
\label{sec:preliminaries}

This appendix collects notation, model conventions, and standard facts used
in the proofs.

\subsection{Notation and queries}
\label{subsec:notation-and-queries}

For a vector, $\|\cdot\|$ denotes its Euclidean norm.  For an operator,
$\|A\|$ is its operator norm, or largest singular value, and
$\|A\|_1=\Tr\sqrt{A^\dagger A}$ is its trace norm, where $\Tr$ denotes
the trace.  A positive semidefinite operator is a Hermitian operator with
nonnegative eigenvalues, and $A\succeq0$ means that $A$ has this property.
A density operator is a positive semidefinite operator of trace one and
represents a possibly mixed quantum state.  The partial trace
$\Tr_B$ discards register $B$.  The conventional trace distance between
states $\rho$ and $\sigma$ is $\tfrac12\|\rho-\sigma\|_1$.
We use the unhalved norm in our estimates.
For a condition $E$, the indicator $\mathbf1_{\{E\}}$ is one when $E$
holds and zero otherwise.

Recall that \(X=\F_2^n\), \(N=\abs X=2^n\), and the group operation
in \(X\) is bitwise XOR \(\oplus\).  For
\(x,y\in X\), let \(x\cdot y\) denote their inner product over \(\F_2\).
For a function $h:X\to X$ and \(D\subseteq X\), the restriction \(h|_D\) is the partial function
\(D\to X\), \(x\mapsto h(x)\).  Saying that an expression depends only on
\(h|_D\) means that changing \(h\) outside \(D\) does not change it.
We write \(\ket a_B\) for the state \(\ket a\) stored in the named quantum
register \(B\).  The subscript labels the register and is not part of \(a\).
Juxtaposed kets denote tensor products, and register labels are omitted when
the intended subsystem is clear.  A ket lists its tensor factors
in the order stated immediately before the formula.  We normally rely on this
order and add a register subscript only when the subsystem would otherwise be
ambiguous.  If \(a\) is a bit string, \(\ket a\) denotes the corresponding
computational-basis state.

For nonzero \(s\in X\), write
\[
  s^\perp=\{y\in X\mid y\cdot s=0\},
  \qquad
  \ket u=\frac1{\sqrt N}\sum_{z\in X}\ket z.
\]
Thus \(\ket u\) is the uniform superposition in one \(n\)-qubit,
\(X\)-valued register.
For a permutation $\tau$ of a finite set $\Omega$, its support is
$\supp(\tau)=\{w\in\Omega\mid\tau(w)\neq w\}$.

The tag table \(r=(r_x)_{x\in X}\) and injection
\(f_{h,r}(x)=(h(x),x,r_x)\) are as in
Definition~\ref{def:garbage-simon}.  We write \(X^X\) for the set of all
such tables.

Between queries, an algorithm may apply arbitrary oracle-independent
quantum channels, introduce ancillas, and perform adaptive measurements.
By purification and deferred measurement, these operations may be treated
as isometries on a sufficiently large workspace~\cite{Watrous2018}.
We include all registers used to purify these operations and store
measurement outcomes in
\(\mathsf{Alg}\), together with the query registers and the algorithm's private
workspace.  Thus $\mathsf{Alg}$ denotes the complete system on which the
algorithm may act.

Let $\mathcal Q:\mathcal K_{\mathrm{in}}\to\mathcal K_{\mathrm{out}}$ be a
query isometry, and let
$J:\mathcal K_{\mathrm{in}}\to\mathcal K_{\mathrm{out}}$
be a fixed oracle-independent isometry.  The \emph{coherent-bypass
extension} of $\mathcal Q$ is
$\mathcal Q^{\mathrm{by}}=\ket0\!\bra0\otimes J+\ket1\!\bra1\otimes\mathcal Q$,
where the first factor is a control qubit.  For an in-place oracle, $J$ may
be the identity.

\begin{proposition}[Coherent bypass]
\label{prop:coherent-bypass}
Suppose that at every query slot the algorithm applies the
coherent-bypass extension of the query, for a fixed oracle-independent
isometry $J$.  Then Corollary~\ref{cor:injective-from-recording},
Theorem~\ref{thm:permutation-lower-bound-estimate}, and
Proposition~\ref{prop:compact} hold with the same bounds.
\end{proposition}

\begin{proof}
If an analytical comparison replaces $\mathcal Q$ by a contraction
$\widetilde{\mathcal Q}$, we use the corresponding direct sum:
\begin{align*}
 \mathcal Q^{\mathrm{by}}
   &=\ket0\!\bra0\otimes J+\ket1\!\bra1\otimes\mathcal Q,\\
 \widetilde{\mathcal Q}^{\mathrm{by}}
   &=\ket0\!\bra0\otimes J
      +\ket1\!\bra1\otimes\widetilde{\mathcal Q},\\
 \|\mathcal Q^{\mathrm{by}}-\widetilde{\mathcal Q}^{\mathrm{by}}\|
   &=\|\mathcal Q-\widetilde{\mathcal Q}\|,
 \qquad
 \|\widetilde{\mathcal Q}^{\mathrm{by}}\|\leq1.
\end{align*}
Here the last bound follows because the two control blocks are orthogonal,
$\|J\|=1$, and $\|\widetilde{\mathcal Q}\|\leq1$.  The same orthogonality,
together with the fact that $J$ and $\mathcal Q$ are isometries, shows that
$\mathcal Q^{\mathrm{by}}$ is an isometry.

For completeness, consider the recording hypotheses and sector projectors
$\Pi_D$ defined in Section~\ref{sec:recording-theorem}.  We extend $J$ by the
identity on the table register $\mathcal R$.  If
$\widetilde{\mathcal Q}^{D'\leftarrow D}$ is an active sector block, then
the corresponding bypass sector block is
\[
 (\widetilde{\mathcal Q}^{\mathrm{by}})^{D'\leftarrow D}
 =\mathbf 1_{\{D'=D\}}\ket0\!\bra0\otimes J\otimes\Pi_D
  +\ket1\!\bra1\otimes
    \widetilde{\mathcal Q}^{D'\leftarrow D}.
\]
Here $\mathbf 1_{\{D'=D\}}$ equals one when $D'=D$ and zero otherwise.
The bypass branch leaves $D$ unchanged, while the active branch retains its
original sector-transition rule.  The active blocks retain their original
locality property, and the bypass blocks are the same for every hidden
function because $J$ is oracle independent.  Thus the approximation bound
retains the same $\delta$, the direct sum remains a contraction, and the
sector-transition and locality conditions hold unchanged.  The fixed-table
condition also survives: on a basis table
$\ket r$, use
$\mathcal Q^{\mathrm{by},(r)}=\ket0\!\bra0\otimes J+
\ket1\!\bra1\otimes\mathcal Q^{(r)}$, whose two branches both leave the
table label unchanged.  If the active comparison appends an auxiliary
register, adjoin the same fixed zero-output register to the bypass map $J$.
For the permutation problem, the lower bound also uses the
completion-to-ideal hybrid of Corollary~\ref{cor:completion-to-ideal}.  A
controlled forward call is represented by the signed
simulation~\eqref{eq:signed-simulation} with both XOR calls and the swap
controlled by the bypass qubit, so the cleanup register still returns to
$\ket0$ in the real computation.  The controlled real and comparison XOR
oracles agree except on control-one branches whose signed input lies in
$\Gamma_W$.  Hence the proof of Lemma~\ref{lem:sparse-disagreement-hybrid}
applies without change, and Corollary~\ref{cor:completion-to-ideal} holds
with the same bound.  Every lower
bound in the paper therefore holds for the coherent-bypass extension.
\end{proof}

In the main proofs we omit the bypass register.

\subsection{Simon sampling}
\label{subsec:simon-sampling}

Recall that a Simon function with nonzero shift $s$ satisfies
$h(x)=h(y)$ exactly when $y\in\{x,x\oplus s\}$, as defined in
Section~\ref{subsec:overview-problems}.

The following standard calculation is the only property of Simon's
algorithm needed in the upper bound~\cite{Simon1997}.  It identifies the
precise distribution of one Simon sample in each promise class.

\begin{lemma}[Simon sampling]
\label{lem:simon-sampling}
We prepare \(N^{-1/2}\sum_x\ket{x}\ket{h(x)}\), measure the second register,
apply the $n$-qubit Hadamard transform to the first, and measure a vector
\(y\in X\).
If \(h\) is a permutation, then \(y\) is uniform on \(X\).  If \(h\) is a
Simon function with shift \(s\), then \(y\) is uniform on \(s^\perp\).
\end{lemma}

\begin{proof}
For a permutation, measuring the second register selects one basis state
\(\ket x\), whose Hadamard measurement is uniform on \(X\).  For a Simon
function, the measurement selects
\(2^{-1/2}(\ket x+\ket{x\oplus s})\).  Its Hadamard transform is
\[
  \sqrt{\frac2N}
  \sum_{\substack{y\in X\\y\cdot s=0}}
  (-1)^{x\cdot y}\ket y,
\]
which gives the stated uniform distribution.
\end{proof}

\subsection{Trace distance and norms}

For an operator \(A\), write
\(\norm{A}_1=\Tr\sqrt{A^\dagger A}\) for its trace norm and \(\norm{A}\) for
its operator norm, or largest singular value.  If \(A\) is positive
semidefinite, then \(\norm{A}_1=\Tr A\).  We use the unhalved trace norm
throughout.  The
conventional trace distance is \(\frac12\norm{\rho_0-\rho_1}_1\).  If
equally likely density operators \(\rho_0,\rho_1\)
are presented, Helstrom's theorem gives optimal discrimination probability
\cite{Watrous2018}
\[
  p_{\rm opt}
  =\frac12+\frac14\norm{\rho_0-\rho_1}_1.
\]
Partial trace, averaging, and quantum channels do not increase the trace
norm of a Hermitian operator.  We also use
\(\norm{\ket a\bra b}_1=\norm{\ket a}\,\norm{\ket b}\).  Consequently, if
\(\norm{\ket\psi}=1\), \(\norm{\ket\phi}\leq1\), and
\(\norm{\ket\psi-\ket\phi}\leq\varepsilon\), then
\begin{equation}
  \norm{
    \ket\psi\bra\psi-\ket\phi\bra\phi
  }_1
  \leq 2\varepsilon.
  \label{eq:pure-to-trace}
\end{equation}
Finally, if an algorithm succeeds with a stated probability on every
promised input, it succeeds with at least that probability after averaging
over any input distribution.  This elementary distributional implication is
the only direction of Yao's principle used in the paper.

\section{Injective problem}

\label{sec:problem}

Recall the injection \(f_{h,r}(x)=(h(x),x,r_x)\) from
Definition~\ref{def:garbage-simon}.  Its tag table is fixed
for the entire oracle instance:
every query to the same \(x\) returns the same tag \(r_x\).  The promise
permits every \(r\), including tables with repeated entries.  Randomness is
used only to define a hard input distribution.

The output registers are ordered as the $h$-value register $H$, the returned
input register $\mathsf X$, and the tag register $Z$.  On its input space the
forward-erasing oracle acts as
\[
   \ket{0^n,x,0^n}_{H,\mathsf X,Z}
      \longmapsto \ket{h(x),x,r_x}_{H,\mathsf X,Z}.
\]
Equivalently,
\[
\begin{quantikz}[row sep={0.75cm,between origins},column sep=0.55cm]
 \lstick{\(H:\ket{0^n}\)} & \gate[wires=3]{E_{f_{h,r}}} & \rstick{\(H:\ket{h(x)}\)} \qw \\
 \lstick{\(\mathsf X:\ket x\)} &                              & \rstick{\(\mathsf X:\ket x\)} \qw \\
 \lstick{\(Z:\ket{0^n}\)} &                                  & \rstick{\(Z:\ket{r_x}\)} \qw
\end{quantikz}
\]
This three-register realization is defined on the subspace in
which $H$ and $Z$ begin in $\ket{0^n}$.

We use the promise classes and bounded-error query conventions from
Section~\ref{sec:overview}.

The upper bounds are proved in Appendix~\ref{sec:upper}.
Section~\ref{sec:injective-lower-bound} proves the lower bound using this
interface.

\section{Injective upper bounds}

\label{sec:upper}

The standard oracle can suppress output coordinates by a choice of target
state.  This makes one query to \(f_{h,r}\) produce exactly one ordinary
Simon sample, independently of the tag table.

\begin{proposition}[Standard-query upper bound]
\label{prop:xor-upper}
The problem \(\mathrm{GarbageSimon}_n\) is solvable with at most \(n+2\)
standard queries.
\end{proposition}

\begin{proof}
In the register order (input register, $H$-target, $\mathsf X$-target, $Z$-target), we prepare
\[
  \frac1{\sqrt N}\sum_{x\in X}
  \ket x\ket{0^n}\ket u\ket u.
\]
A query to \(O_{f_{h,r}}\) produces
\[
  \frac1{\sqrt N}\sum_{x\in X}
  \ket x\ket{h(x)}\ket u\ket u,
\]
because XOR translation merely permutes the basis labels in \(\ket u\) and
therefore leaves it invariant.  Lemma
\ref{lem:simon-sampling} therefore gives one sample from \(X\) in the
permutation case and from \(s^\perp\) in the Simon case.

We take \(k=n+2\) independent samples and output ``permutation'' exactly when
they span \(X\).  Simon samples lie in the proper subspace \(s^\perp\), so
they never span \(X\).  For uniform samples from \(X\), failure to span
implies that some nonzero vector is orthogonal to all samples.  A union
bound gives
\[
  \Pr[\text{failure to span}]
  \leq (2^n-1)2^{-k}<\frac14.
\]
Thus the worst-case success probability is greater than \(3/4\).
\end{proof}

The matching forward-erasing upper bound is classical collision sampling:
by the birthday paradox, $O(\sqrt N)$ random queries find a Simon pair with
constant probability.
It is included to identify the correct scale and to make the final
\(\Theta(\sqrt N)\) statement self-contained.

\begin{proposition}[Forward-erasing upper bound]
\label{prop:erase-upper}
The problem \(\mathrm{GarbageSimon}_n\) is solvable with \(O(\sqrt N)\)
forward-erasing queries.
\end{proposition}

\begin{proof}
We choose \(k\) distinct inputs uniformly, query them in the computational
basis, and measure the register containing \(h(x)\).  The complete outputs of
\(f_{h,r}\) never collide, because they include \(x\).  This test compares
  only the values in the \(H\) register.  A permutation produces no
  collision there.
For a Simon function, the input set is partitioned into
\(N/2\) Simon pairs.  When \(k\leq N/2\), the probability that a uniform
\(k\)-subset contains no complete pair is
\[
  \frac{2^k\binom{N/2}{k}}{\binom Nk}
  =\prod_{j=0}^{k-1}\frac{N-2j}{N-j}
  \leq
  \exp\left(-\frac{k(k-1)}{2N}\right).
\]
For \(N\geq16\), we take \(k=\lceil2\sqrt N\rceil\).  The last expression is at
most \(e^{-7/4}<1/3\).  For \(N<16\), we query all \(N\leq4\sqrt N\) inputs.
\end{proof}

\section{Sector induction}
\label{app:sector-induction}

This appendix gives the sector-support and locality induction used in
Theorem~\ref{thm:local-recording}.  We first illustrate its local
distribution comparison on four addresses.

\medskip
\noindent\emph{Four-point example.}
Let $X=\{\mathtt{00},\mathtt{01},\mathtt{10},\mathtt{11}\}$ and fix the
two-address set $D=\{\mathtt{00},\mathtt{01}\}$.  Under $\Hcal_1$, the shift
is uniform over $\{\mathtt{01},\mathtt{10},\mathtt{11}\}$.  If
$s\in\{\mathtt{10},\mathtt{11}\}$, the two points of $D$ lie in distinct
Simon pairs, so $h|_D$ is a uniformly random injection, exactly as under
$\Hcal_0$.  If $s=\mathtt{01}$, they form one Simon pair and their values are
equal, with the common value uniform in $X$.  Thus the two restriction
distributions differ for exactly one of the three possible shifts.  The
exceptional probability is $1/3$, which equals
$\binom{|D|}{2}/(N-1)$ here.  This is a concrete instance of the
comparison of restriction distributions used below.

The induction below uses only
Conditions~\ref{item:recording-sector} and~\ref{item:recording-local}.
As in the theorem, all hidden functions compared below belong to
$\supp(\Hcal_0)\cup\supp(\Hcal_1)$.

\begin{proof}[Proof of the sector and locality claims]
Fix $\theta$.  For $0\leq\ell\leq T$, let
$\ket{\widetilde\Psi^{(\ell)}_{h,\theta}}$ denote the complete global vector
after $\ell$ applications of the recording contraction, including the
oracle-independent isometry following query $\ell$ when $\ell\geq1$.  The initial
oracle-independent preparation is included in the vector at $\ell=0$,
before any query.  Let $\mathsf S_\ell$ denote all non-table registers
at the time defining $\ket{\widetilde\Psi^{(\ell)}_{h,\theta}}$, and define
\[
   \ket{\widetilde\Psi^{(\ell)}_{h,\theta,D}}
      =
      (\Id_{\mathsf S_\ell}\otimes\Pi_D)
      \ket{\widetilde\Psi^{(\ell)}_{h,\theta}}.
\]

We prove simultaneously that
\begin{equation*}
   \ket{\widetilde\Psi^{(\ell)}_{h,\theta}}
      =
      \sum_{\substack{D\subseteq X\\|D|\leq\ell}}
         \ket{\widetilde\Psi^{(\ell)}_{h,\theta,D}}
   \tag{$A_\ell$}
\end{equation*}
and that, for every $D\subseteq X$,
\begin{equation*}
   h|_D=h'|_D
   \quad\Longrightarrow\quad
   \ket{\widetilde\Psi^{(\ell)}_{h,\theta,D}}
      =
   \ket{\widetilde\Psi^{(\ell)}_{h',\theta,D}}.
   \tag{$L_\ell$}
\end{equation*}
The first assertion says that at most $\ell$ cells are recorded after
$\ell$ queries.  The second states locality as equality of complete sector
vectors, including every register retained by the comparison computation.

The sector projectors resolve the identity:
\[
   \sum_{D\subseteq X}\Pi_D
      =
   \bigotimes_{i\in X}(\Pi_{0,i}+\Pi_{1,i})
      =
   \Id_{\mathcal R}.
\]
Thus $(A_\ell)$ is equivalent to the vanishing of every component with
$|D|>\ell$.

Before the first query, the retained vector is independent of $h$, and the
table register is in
\[
   \ket U=\bigotimes_{i\in X}\ket u.
\]
Since $\Pi_{0,i}\ket u=\ket u$ and $\Pi_{1,i}\ket u=0$, the table vector
lies entirely in $\mathcal R_\varnothing$.  Hence $(A_0)$ holds.  The
empty-sector vector is independent of $h$, and every other sector vector is
zero, so $(L_0)$ also holds.

Suppose $(A_\ell)$ and $(L_\ell)$ hold.  Let
$\mathcal U_{\ell+1}$ be the oracle-independent isometry following query
slot $\ell+1$, with its identity action on $\mathcal R$ understood.  At this
step, $\widetilde{\mathcal Q}_{h,\theta}$ denotes the recording contraction
for slot $\ell+1$.  As in the theorem statement, the slot index is
suppressed.  Let $\mathsf T_{\ell+1}$ be the retained system immediately
after that query and before $\mathcal U_{\ell+1}$.  Although the retained
input and output spaces of $\mathcal U_{\ell+1}$ may differ, its identity
action on $\mathcal R$ preserves every sector label:
\[
 (\Id_{\mathsf S_{\ell+1}}\otimes\Pi_{D'})
 (\mathcal U_{\ell+1}\otimes\Id_{\mathcal R})
 =
 (\mathcal U_{\ell+1}\otimes\Id_{\mathcal R})
 (\Id_{\mathsf T_{\ell+1}}\otimes\Pi_{D'}).
\]
The component in an output sector $D'$ is therefore
\begin{equation}
\begin{aligned}
   \ket{\widetilde\Psi^{(\ell+1)}_{h,\theta,D'}}
      &=
      \sum_{\substack{D\subseteq X\\|D|\leq\ell}}
      \mathcal U_{\ell+1}
      \widetilde{\mathcal Q}_{h,\theta}^{D'\leftarrow D}
      \ket{\widetilde\Psi^{(\ell)}_{h,\theta,D}}.
\end{aligned}
\label{eq:sector-induction-recurrence}
\end{equation}
Here $(A_\ell)$ supplies the predecessor sectors, and inserting the input
and output sector projectors gives the blocks in
Condition~\ref{item:recording-sector}.

Every nonzero summand in
\eqref{eq:sector-induction-recurrence} satisfies
\[
   D\subseteq D',
   \qquad
   |D'\setminus D|\leq1.
\]
Since $|D|\leq\ell$, such a summand can occur only when
$|D'|\leq\ell+1$.  This proves $(A_{\ell+1})$.

For locality, fix $D'$ and suppose $h|_{D'}=h'|_{D'}$.  Only predecessors
$D\subseteq D'$ can contribute.  For each such $D$, the two functions also
agree on $D$, so $(L_\ell)$ gives
\[
   \ket{\widetilde\Psi^{(\ell)}_{h,\theta,D}}
      =
   \ket{\widetilde\Psi^{(\ell)}_{h',\theta,D}}.
\]
Condition~\ref{item:recording-local} gives
\[
   \widetilde{\mathcal Q}_{h,\theta}^{D'\leftarrow D}
      =
   \widetilde{\mathcal Q}_{h',\theta}^{D'\leftarrow D}.
\]
Finally, $\mathcal U_{\ell+1}$ is independent of $h$.  Every summand in
\eqref{eq:sector-induction-recurrence}, and hence their coherent sum, is
therefore unchanged when $h$ is replaced by $h'$.  This proves
$(L_{\ell+1})$.

Taking $\ell=T$ yields the sector decomposition and locality statement used
in the proof of Theorem~\ref{thm:local-recording}.
\end{proof}

\section{Gram identities}
\label{app:gram-identities}

We derive the four identities used in the contraction part of
Lemma~\ref{lem:permutation-one-query-recording}.  Fix $i\in X$ and, for
$t\in X$, let $y_{i,t}=(1,h(i),i,t)$ and $w_t=\pi^{-1}(y_{i,t})$.  We define
\[
   \mathcal C_i
      =\{t\in X\mid w_t\notin\mathcal A\},
   \qquad
   \ket{p_t}=\Pi_{1,i}\ket t
   \quad\text{for }t\in\mathcal C_i.
\]
Let $d_t=w_t\oplus a_i$.  Since $w_t\notin\mathcal A$ and
$a_i\in\mathcal A$, we have $d_t\neq0$, so
$\langle0|d_t\rangle_G=0$.  We suppress the identity on all table cells
other than $R_i$.  On the selected cell, the relevant unmodified maps act
as
\[
   M_{a_i}\ket z
      =\ket{y_{i,z}}_Y\ket0_G\ket z_{R_i},
   \qquad
   M_{w_t}^{\mathrm{base}}\ket z
      =\ket{y_{i,t}}_Y\ket0_G\ket z_{R_i},
\]
and
\[
   M_{w_t}\ket z
      =
      \begin{cases}
         \ket{y_{i,t}}_Y\ket{d_t}_G\ket t_{R_i},
            &z=t,\\[1mm]
         \ket{y_{i,t}}_Y\ket0_G\ket z_{R_i},
            &z\neq t.
      \end{cases}
\]
The two branch maps of the recording contraction in the $i$-block are
\[
   V_i=\widetilde M_{a_i}=\Pi_{1,i}M_{a_i},
   \qquad
   K_{i,t}=\widetilde M_{w_t}
      =M_{w_t}^{\mathrm{base}}\Pi_{0,i}
        +\Pi_{1,i}M_{w_t}\Pi_{1,i}.
\]

We begin with products of the unmodified maps.  We write $\delta_{z,z'}=1$
if $z=z'$ and $\delta_{z,z'}=0$ otherwise.  For $z,z'\in X$,
\begin{align*}
 \bra z M_{a_i}^{\dagger}\Pi_{0,i}M_{a_i}\ket{z'}
   &=
   \langle y_{i,z}|y_{i,z'}\rangle
   \langle0|0\rangle_G
   \bra z\Pi_{0,i}\ket{z'} \\
   &=\frac1N\delta_{z,z'}.
\end{align*}
Therefore
\begin{equation}
   M_{a_i}^{\dagger}\Pi_{0,i}M_{a_i}
      =\frac1N\Id.
   \label{eq:app-special-product}
\end{equation}

For the candidate branch, the main output is always $y_{i,t}$.  After the
output table cell is projected with
$\Pi_{0,i}=\ket u\!\bra u$, two basis inputs have nonzero overlap exactly
when their $G$ labels agree.  Consequently,
\[
 \bra zM_{w_t}^{\dagger}\Pi_{0,i}M_{w_t}\ket{z'}
 =
 \begin{cases}
   \dfrac1N,
      &z=z'=t,\\[1mm]
   \dfrac1N,
      &z\neq t\text{ and }z'\neq t,\\[1mm]
   0,
      &\text{exactly one of $z,z'$ equals $t$}.
 \end{cases}
\]
Equivalently,
\begin{equation}
 M_{w_t}^{\dagger}\Pi_{0,i}M_{w_t}
 =
 \frac1N\left[
   \ket t\!\bra t
   +
   \left(\sum_{z\neq t}\ket z\right)
   \left(\sum_{z'\neq t}\bra{z'}\right)
 \right].
 \label{eq:app-candidate-product}
\end{equation}
Since
\[
   \Pi_{1,i}\ket t=\ket{p_t},
   \qquad
   \Pi_{1,i}\sum_{z\neq t}\ket z
      =\Pi_{1,i}\bigl(\sqrt N\ket u-\ket t\bigr)
      =-\ket{p_t},
\]
sandwiching~\eqref{eq:app-candidate-product} between $\Pi_{1,i}$ on both
sides gives
\begin{equation}
 \Pi_{1,i}M_{w_t}^{\dagger}\Pi_{0,i}M_{w_t}\Pi_{1,i}
    =\frac2N\ket{p_t}\!\bra{p_t}.
 \label{eq:app-candidate-defect}
\end{equation}

For the mixed product, the main outputs of $M_{a_i}\ket z$ and
$M_{w_t}\ket{z'}$ agree only when $z=t$.  In that case their $G$ labels
agree only when $z'\neq t$.  Hence
\[
 M_{a_i}^{\dagger}\Pi_{0,i}M_{w_t}
 =
 \frac1N\ket t\left(\sum_{z'\neq t}\bra{z'}\right).
\]
The identity
\[
   \left(\sum_{z'\neq t}\bra{z'}\right)\Pi_{1,i}
      =-\bra{p_t}
\]
then yields
\begin{equation}
 M_{a_i}^{\dagger}\Pi_{0,i}M_{w_t}\Pi_{1,i}
   =-\frac1N\ket t\!\bra{p_t}.
 \label{eq:app-mixed-product}
\end{equation}

We also need
\begin{equation}
   M_{a_i}^{\dagger}M_{w_t}=0.
   \label{eq:app-unmodified-orthogonality}
\end{equation}
To verify it, first consider unequal table-basis inputs.  Their unchanged
table labels are orthogonal.  For equal inputs $z=z'\neq t$, the main
outputs $y_{i,z}$ and $y_{i,t}$ are distinct.  For $z=z'=t$, the main
outputs agree, but the $G$ labels are $0$ and $d_t$, which are distinct.

We can now establish the four Gram identities.  First,
using~\eqref{eq:app-special-product} and the fact that $M_{a_i}$ is an
isometry,
\begin{align*}
 V_i^\dagger V_i
   &=M_{a_i}^{\dagger}\Pi_{1,i}M_{a_i} \\
   &=M_{a_i}^{\dagger}
       (\Id-\Pi_{0,i})M_{a_i} \\
   &=\left(1-\frac1N\right)\Id.
\end{align*}

Second, the two summands of $K_{i,t}$ have output table factors in the
orthogonal subspaces $\operatorname{im}\Pi_{0,i}$ and
$\operatorname{im}\Pi_{1,i}$.  Moreover,
$M_{w_t}^{\mathrm{base}}$ is an isometry and acts trivially on the table.
It follows from~\eqref{eq:app-candidate-defect} that
\begin{align*}
 K_{i,t}^{\dagger}K_{i,t}
   &=
   \Pi_{0,i}
   +
   \Pi_{1,i}M_{w_t}^{\dagger}
      \Pi_{1,i}M_{w_t}\Pi_{1,i} \\
   &=
   \Pi_{0,i}+\Pi_{1,i}
   -
   \Pi_{1,i}M_{w_t}^{\dagger}
      \Pi_{0,i}M_{w_t}\Pi_{1,i} \\
   &=
   \Id-\frac2N\ket{p_t}\!\bra{p_t}.
\end{align*}

Third, the baseline part of $K_{i,t}$ has output table factor in
$\operatorname{im}\Pi_{0,i}$, so it is killed by the $\Pi_{1,i}$ in
$V_i^\dagger$.  Using~\eqref{eq:app-mixed-product} and
\eqref{eq:app-unmodified-orthogonality},
\begin{align*}
 V_i^\dagger K_{i,t}
   &=
   M_{a_i}^{\dagger}\Pi_{1,i}M_{w_t}\Pi_{1,i} \\
   &=
   M_{a_i}^{\dagger}M_{w_t}\Pi_{1,i}
   -
   M_{a_i}^{\dagger}\Pi_{0,i}M_{w_t}\Pi_{1,i} \\
   &=
   \frac1N\ket t\!\bra{p_t}.
\end{align*}
The sign in the last line is positive because the subtraction in the second
line is applied to the negative operator in
\eqref{eq:app-mixed-product}.

Finally, every vector in the range of $K_{i,t}$ has its $Y$ register in the
basis state $\ket{y_{i,t}}$.  For $t\neq t'$, these basis states are
orthogonal, so
\[
   K_{i,t}^{\dagger}K_{i,t'}=0.
\]
These calculations prove all four identities in the main text.  Tensoring
the maps with the identity on $\mathcal R_{\ne i}$ and on the algorithm
workspace does not change them.

\section{Contraction calculation}
\label{app:contraction-calculation}

We now use the Gram identities to prove the contraction estimate on an
arbitrary coherent input, including interference between the special input
$a_i$ and all candidate background inputs associated with $i$.

An arbitrary input supported on the $i$-block has the form
\[
   \ket{\Phi_{\mathrm{in}}}
      =
      \ket{a_i}_Y\ket\gamma
      +
      \sum_{t\in\mathcal C_i}
         \ket{w_t}_Y\ket{\zeta_t},
\]
where $\ket\gamma$ and the $\ket{\zeta_t}$ are arbitrary vectors on the
table and all spectator registers.  The input basis states in $Y$ are
orthogonal, so
\[
   \|\Phi_{\mathrm{in}}\|^2
      =
      \|\gamma\|^2
      +
      \sum_{t\in\mathcal C_i}\|\zeta_t\|^2.
\]
The output of the recording map is
\[
   \ket{\Phi_{\mathrm{out}}}
      =
      V_i\ket\gamma
      +
      \sum_{t\in\mathcal C_i}
         K_{i,t}\ket{\zeta_t}.
\]
Expanding the squared norm and inserting the four identities from
Appendix~\ref{app:gram-identities} gives
\begin{align*}
 \|\Phi_{\mathrm{out}}\|^2
   &=
   \left(1-\frac1N\right)\|\gamma\|^2
   +
   \sum_{t\in\mathcal C_i}\|\zeta_t\|^2 \\
   &\quad
   -\frac2N
      \sum_{t\in\mathcal C_i}
      \left\|
         (\bra{p_t}\otimes\Id)\ket{\zeta_t}
      \right\|^2 \\
   &\quad
   +\frac2N\operatorname{Re}
      \left\langle
         \gamma,
         \sum_{t\in\mathcal C_i}
         (\ket t\!\bra{p_t}\otimes\Id)
         \ket{\zeta_t}
      \right\rangle.
\end{align*}
There are no off-diagonal candidate--candidate terms because
$K_{i,t}^{\dagger}K_{i,t'}=0$ for $t\neq t'$.

We define
\[
   c
      =
      \sum_{t\in\mathcal C_i}
      (\ket t\!\bra{p_t}\otimes\Id)\ket{\zeta_t}.
\]
If
\[
   \ket{\chi_t}
      =(\bra{p_t}\otimes\Id)\ket{\zeta_t},
\]
then $c=\sum_t\ket t\ket{\chi_t}$.  The computational-basis states
$\ket t$ are orthonormal, even though the vectors $\ket{p_t}$ need not be.
Therefore
\[
   \|c\|^2
      =
      \sum_{t\in\mathcal C_i}\|\chi_t\|^2
      =
      \sum_{t\in\mathcal C_i}
      \left\|(\bra{p_t}\otimes\Id)\ket{\zeta_t}\right\|^2.
\]
Subtracting the input norm from the output norm and completing the square
now gives
\begin{align*}
 \|\Phi_{\mathrm{out}}\|^2-\|\Phi_{\mathrm{in}}\|^2
   &=
   -\frac1N\|\gamma\|^2
   -\frac2N\|c\|^2
   +\frac2N\operatorname{Re}\langle\gamma,c\rangle \\
   &=
   -\frac1N
      \bigl(\|\gamma-c\|^2+\|c\|^2\bigr)
   \leq0.
\end{align*}
Thus $\widetilde M$ is a contraction on every $i$-block.

It remains to combine the blocks.  Blocks with different indices $i$ have
orthogonal $Y$ supports because the address field is part of $y_{i,t}$.
A remaining background branch cannot have $Y$ output equal to an output in
a candidate block: if $\pi(w)=y_{i,t}$ for a background input $w$, then it
is the candidate branch indexed by $i,t$.  On every remaining background
branch $\widetilde M=M$, and distinct inputs have distinct $Y$ outputs
because $\pi$ is a permutation.  The complete input space therefore
decomposes into orthogonal blocks whose output spaces are also orthogonal.
Since the map is contractive on each block,
\[
   \|\widetilde M\|\leq1.
\]

\section{Permutation locality}
\label{app:permutation-locality}

We give the full sector-block verification for
Lemma~\ref{lem:permutation-locality}.  In this section we write
$\widetilde M_{h,\pi}$ when the dependence of the recording map on $h$ and
$\pi$ needs to be explicit.  For a basis input $w$, let
$\widetilde M_{h,\pi,w}:=\widetilde M_{h,\pi}\iota_w$ denote its branch map,
in the notation of Section~\ref{subsec:recording-contraction}.

For one query slot, define
\[
   \widetilde M_{h,\pi}^{D'\leftarrow D}
      =
      (\Id_{\mathsf S_{\rm out}}\otimes\Pi_{D'})
      \widetilde M_{h,\pi}
      (\Id_{\mathsf S_{\rm in}}\otimes\Pi_D).
\]
We first prove the following two one-query facts:
\begin{align*}
 \widetilde M_{h,\pi}^{D'\leftarrow D}=0
   &\quad\text{unless}\quad
   D\subseteq D'\text{ and }|D'\setminus D|\leq1,
   \tag{A}\label{eq:app-permutation-transition}\\
 h|_{D'}=h'|_{D'}
   &\quad\Longrightarrow\quad
   \widetilde M_{h,\pi}^{D'\leftarrow D}
      =\widetilde M_{h',\pi}^{D'\leftarrow D}.
   \tag{B}\label{eq:app-permutation-block-locality}
\end{align*}
Because the query-input states $\ket w_Y$ form a basis, it suffices to
verify both claims on each input branch.

We first consider the special input $w=a_i$.  By
\eqref{eq:Mtilde-special}, its output table component lies in
$\operatorname{im}\Pi_{1,i}$, while every other table cell is unchanged.
It can therefore take sector $D$ only to $D\cup\{i\}$.  This either leaves
the sector unchanged, when $i\in D$, or adds the single index $i$.  The
branch depends on $h$ only through $h(i)$, and $i$ belongs to the output
sector in either case.  Thus, if this branch contributes to a block ending
in $D'$ and $h|_{D'}=h'|_{D'}$, then $h(i)=h'(i)$ and the two branch maps
are equal.

Now let $w\notin\mathcal A$ and set $y=\pi(w)$.  If
$y\notin\Omega_1$, the map is the table-independent baseline and preserves
every sector.  Suppose instead that $y=(1,v,i,t)$.  Its classification as a
candidate branch is determined by whether $v=h(i)$.  Two cases cover every
input sector.

If $i\notin D$, then $\Pi_D$ contains $\Pi_{0,i}$.  In the candidate case,
the rightmost $\Pi_{1,i}$ in the second summand of
\eqref{eq:Mtilde-candidate} kills that input component, while the first
summand acts as the baseline:
\[
   \widetilde M_{h,\pi,w}\Pi_D
      =M_w^{\mathrm{base}}\Pi_D.
\]
The same equality holds in the noncandidate case by definition.  Thus the
action on sector $D$ is independent of $h(i)$ even if changing $h(i)$
changes whether the branch is called a candidate.  It also preserves the
table sector, so only $D'=D$ can occur.

If $i\in D$, then $\Pi_D$ contains $\Pi_{1,i}$.  A candidate branch acts
through $\Pi_{1,i}M_w\Pi_{1,i}$ and hence stays in the recorded subspace of
cell $i$.  A noncandidate branch uses the baseline and also preserves that
subspace.  All other cells are unchanged, so again $D'=D$.  The two actions
can differ, but this causes no violation of locality: if
$h|_D=h'|_D$, then $h(i)=h'(i)$, so the branch has the same classification
and the same action for both functions.

This branchwise analysis proves
\eqref{eq:app-permutation-transition} and
\eqref{eq:app-permutation-block-locality}.  It also isolates the reason for
the baseline term in~\eqref{eq:Mtilde-candidate}: on an input sector not
containing $i$, candidate and noncandidate branches must act identically.

The global decomposition and locality claim now follow from the induction
in Appendix~\ref{app:sector-induction}, with $\theta=\pi$ and
$\widetilde{\mathcal Q}_{h,\theta}=\widetilde M_{h,\pi}$.  The table starts
in the empty sector, and the inter-query operations act as the identity
on $\mathcal R$.  Equations~\eqref{eq:app-permutation-transition}
and~\eqref{eq:app-permutation-block-locality} verify exactly the two
hypotheses used in that induction.  After $T$ queries, only sectors of
size at most $T$ occur, and the complete vector in sector $D$ depends
only on $h|_D$, proving Lemma~\ref{lem:permutation-locality}.

\paragraph{Acknowledgement.} This work was supported by Tamkeen under the NYUAD Research Institute grant CG008.

\section*{AI methodology}
The authors used ChatGPT iteratively in developing the results presented
in this paper.  For the permutation separation, ChatGPT was used
extensively to explore alternative approaches before settling on the
proof presented here.  The authors corrected, simplified, and
substantially rewrote the ideas in the initial draft.  ChatGPT and Claude
assisted with expository revisions and organization, and Claude was used
to review the entire paper, including the references.  The authors
verified the correctness and originality of all content, including
references, and take full responsibility for the results.

\clearpage
\begingroup
\apptocmd{\thebibliography}{\setlength{\itemsep}{1pt}\setlength{\parsep}{0pt}}{}{}
\bibliographystyle{alpha}
\bibliography{references}
\endgroup

\end{document}